\documentclass[11pt,a4paper]{article}

\usepackage{amsmath}
\usepackage{amssymb}
\usepackage{common}
\usepackage{thmtools}

\usepackage{tikz}
\usetikzlibrary{decorations.pathmorphing}
\usepackage{tikz-cd}
\usepackage{mathpartir}

\tikzset{
    tripledata/.style args={[#1] in [#2] in [#3]}{
        #1,preaction={preaction={draw,#3},draw,#2}
    }
}
\tikzset{triple/.style={tripledata={[line width=.4pt,black] in
      [line width=3pt,white] in
      [line width=3.8pt,black]}}}

\tikzset{squiggly/.style={decorate, decoration={snake, segment length=5pt, amplitude=1pt}}}

\tikzset{standard/.style={decorate, decoration={snake, segment length=4.6pt, amplitude=1pt}, -{>[length=3pt,width=6pt]}, line width=1.2pt, line cap=round}}
\tikzset{nonstandard/.style={decorate, decoration={snake, segment length=4.6pt, amplitude=1pt}, -{>[length=2pt,width=4pt]}, line width=0.4pt, line cap=round}}
\tikzset{point/.style={draw, circle, black, fill, inner sep=1.5pt, label={}}}

\newcommand\cat[1]{\ensuremath{\mathsf{#1}}\xspace}

\newcommand{\Hom}{\mathsf{Hom}}

\newcommand{\src}{\mathsf{src}}
\newcommand{\tgt}{\mathsf{tgt}}
\newcommand{\sub}[1]{[ \hspace{-2.1pt} [ #1 ] \hspace{-2.1pt} ]}
\newcommand{\ty}{\mathsf{ty}}

\renewcommand{\dim}{\mathsf{dim}}

\renewcommand{\sslash}{\mathord{/\mkern-5mu/}}
\newcommand\squig{\mathrel{\tikz{
  \draw [decorate, decoration={snake, segment length=4.75pt, amplitude=1pt}, -{>[length=2pt,width=4pt]}, line cap=round] (-.02,0) to +(.266,0);
}}}
\newcommand\xsquig[1]{\overset{\scriptstyle{\mathrm{#1}}}{\squig}}
\newcommand\bsquig{\mathrel{\begin{tikzpicture}
   \draw [decorate, decoration={snake, segment length=5pt, amplitude=1pt}, -{>[length=3pt,width=5pt]}, line width=1.1pt, line cap=round] (-.033,0) to +(.28,0);
\end{tikzpicture}}}
\newcommand\xbsquig[1]{\overset{\scriptstyle{\mathrm{#1}}}{\bsquig}}
\renewcommand\r{\ensuremath{\mathsf{r}}}
\renewcommand\t{\ensuremath{\mathsf{t}}}
\newcommand\rt{\ensuremath{\mathsf{rt}}}
\newcommand\s{\ensuremath{\mathsf{s}}}
\newcommand\rts{\ensuremath{\mathsf{rts}}}

\newcommand\comp{\ensuremath{\mathsf{comp}}}
\newcommand\uD{\ensuremath{\mathrm{D}}}

\newcommand\op{\mathrm{op}}

\newcommand\tm{\mathsf{tm}}

\newcommand\proofstep{

\smallskip
\noindent
\,\,$\bullet$\,\,}

\newcommand{\var}{\mathsf{var}}
\newcommand{\coh}{\ensuremath{\mathsf{coh}\,}}

\newcommand{\Dyck}{\mathsf{Dyck}}

\newcommand{\Peak}{\mathsf{Peak}}

\newcommand{\Type}{\mathsf{Type}}
\newcommand{\Term}{\mathsf{Term}}
\newcommand{\Sub}{\mathsf{Sub}}
\newcommand{\Ctx}{\mathsf{Ctx}}

\newcommand{\FV}{\mathsf{FV}}
\newcommand{\supp}{\mathsf{supp}}
\newcommand{\LM}{\mathsf{LM}}

\newcommand\N{\ensuremath{\mathbb{N}}}
\newcommand\id{\ensuremath{\mathrm{id}}}
\renewcommand\i{\ensuremath{\mathbb{1}}}
\renewcommand\N{\ensuremath{\mathbb{N}}}

\newcommand\Catt{\ensuremath{\mathsf{Catt}}\xspace}
\newcommand\cCat{\mathcal{C}\hspace{-.9pt}at}
\newcommand\su{\ensuremath{\mathsf{su}}}

\newcommand\pd{\ensuremath{\mathsf{pd}}}

\newcommand\Cattsupd{\ensuremath{\smash{\Catt_\su^\pd}}\xspace}

\mathchardef\mhyphen="2D

\usepackage[margin=3cm]{geometry}
\setlist[itemize]{leftmargin=23pt}

\declaretheorem{theorem}
\declaretheorem[sibling=theorem]{lemma}

\declaretheorem[sibling=theorem,style=definition]{definition}
\declaretheorem[sibling=theorem,style=definition]{example}

\usepackage{etoolbox}
\makeatletter
\def\calign@preamble{%
   &\hfil\strut@
    \setboxz@h{\@lign$\m@th\displaystyle{##}$}%
    \ifmeasuring@\savefieldlength@\fi
    \set@field
    \hfil
    \tabskip\alignsep@
}
\let\cmeasure@\measure@
\patchcmd\cmeasure@{\divide\@tempcntb\tw@}{}{}{}
\patchcmd\cmeasure@{\divide\@tempcntb\tw@}{}{}{}
\patchcmd\cmeasure@{\ifodd\maxfields@
  \global\advance\maxfields@\@ne
  \fi}{}{}{}
\newenvironment{calign}
{%
  \let\align@preamble\calign@preamble
  \let\measure@\cmeasure@
  \align
}
{%
  \endalign
}
\makeatother

\begin{document}

\title{\bf A Type Theory for Strictly Unital \texorpdfstring{$\infty$-}{Infinity }Categories}

\author{Eric Finster\footnote{University of Birmingham, \texttt{e.l.finster@bham.ac.uk}}~,
David Reutter\footnote{University of Hamburg, \texttt{david.reutter@uni-hamburg.de}}~,
Alex Rice\footnote{University of Cambridge, \texttt{alex.rice@cl.cam.ac.uk}}~
and Jamie Vicary\footnote{University of Cambridge, \texttt{jamie.vicary@cl.cam.ac.uk}}}

\maketitle

\begin{abstract}
  We use type-theoretic techniques to present an
  algebraic theory of {$\infty$-categories with strict units}.
  Starting with a known type-theoretic presentation of fully weak $\infty$\-categories, in which terms denote valid operations, we extend the theory with a non-trivial definitional equality. This forces some operations to coincide strictly in any model, yielding the strict unit behaviour.

We make a
  detailed investigation of the meta-theoretic properties of this
  theory. We give a reduction relation that generates definitional equality, and prove that it is confluent and terminating, thus yielding the first decision procedure for equality in a strictly-unital setting.
Moreover, we show that our definitional equality relation identifies all terms in a disc context, providing a point
  comparison with a previously proposed definition of strictly unital
  $\infty$-category. We also prove a conservativity result, showing
  that every operation of the strictly unital theory indeed arises
  from a valid operation in the fully weak theory. From this, we infer  that strict unitality is a property of an $\infty$\-category rather than additional structure.
\end{abstract}

\section{Introduction}

\paragraph{Overview}
 One characteristic feature of logical systems
based on type theory, and one might say constructive mathematics more
generally, is \emph{proof-relevance}.  This is to say that these
systems manipulate mathematical proofs as first-class entities, in
contrast with classical systems which regard proofs as part of
meta-mathematics.  When applied to logical connectives like
implication, conjunction, disjunction, and universal and existential
quantification, proof-relevant principles lead to the connection
between constructive mathematics and programming languages via the
propositions-as-types paradigm.

When applied to \emph{equality}, however, the notion of proof
relevance is considerably more subtle.  For example, the identity
types proposed by Martin-L\"of in his eponymous type theory sometimes
behave in ways which can be surprising and unintuitive to newcomers.
When working in such theories, one quickly arrives at
situations where, given two proofs that two elements are equal, one is
faced with the task of showing that \emph{these two proofs are
  themselves equal}, a task which may seem unusual to a classically
trained mathematician.  Moreover, there are even models where such equality proofs can be unequal~\cite{hofmann1998groupoid}. We say in this case that the principle of \emph{uniqueness of identity
  proofs} fails. Other principles which are often taken for granted
classically, such as \emph{function extensionality} and the existence
of quotients, may also unavailable.

But the subtleties of proof-relevant equality are not confined to
computer science. Indeed, classical mathematicians have been
struggling with them for decades, if from a less foundational point of
view.  Topologists, for example, quickly realized that it was useful
to consider algebraic structures on topological spaces where the
axioms held not at the point-set level, but were rather represented as
paths in the underlying space.  The study of those structures and
properties of spaces which depend only on the algebra of paths is
called \emph{homotopy theory}~\cite{whitehead2012elements, boardman2006homotopy}.

In working with these structures, mathematicians came to discover a
major difficulty that arises when handling algebraic objects where
equality is given as data, rather than a mere property.  When the
axioms of an algebraic structure constitute additional structure, then
this structure must \emph{itself} be subject to ``higher'' axioms,
colloquially known as \emph{coherence data}.  In a proof-relevant
setting, this coherence data will itself generate further coherence
data, in a tower of increasing complexity, yielding a combinatorial
explosion.  This is difficult to handle, and it is a non-trivial
problem to even \textit{describe} well-behaved algebraic theories in
proof-relevant settings.

This difficulty is perhaps best exemplified by the search for a
well-behaved theory of higher dimensional category theory~\cite{leinster}.
By now this is an active area of research, with many
known techniques and proposed definitions.  Moreover, the theory has
applications in areas as diverse as topology, manifold theory,
representation theory, algebra and even fundamental
physics~\cite{baez2010physics}.

The realization that the techniques and ideas which apply in the study
of homotopy theory and higher category theory apply equally well to
the proof-relevant versions of equality found in computer science is
the central intuition of \emph{Homotopy Type Theory}~\cite{hottbook}.
And indeed, importing ideas from homotopy theory has proven fruitful
for understanding the proof-relevant equality of type theory, leading
to new ideas and resolutions of some of the problems described
above~\cite{cohen2016cubical}.

But conversely, ideas from logic and computer science may also be
applied to the description and manipulation of higher dimensional
structures. While classically such structures are often developed from
\emph{combinatorial} descriptions of their underlying data (for
example, using \emph{simplices}~\cite{lurie2009higher}), the
connections with logic suggest an alternative \emph{syntactic}
approach, closer to that of universal algebra.  The problem of
presenting higher-dimensional categories, for example, has recently
been given a syntactic description in the type theory
\Catt~\cite{Finster2017}. This yields an explicit machine-checkable syntax for higher-categorical operations, which applies in principle in arbitrary dimensions.

But beyond the initial problem of describing proof-relevant structures
lies a deeper one: actually using and the resulting definitions in
practice.  Fans of computer formalization can surely attest to the
fact that proof-relevant theories can sometimes be extremely verbose.  This can
be especially frustrating, because it is often intuitively clear that
much of the data ought to be \emph{redundant}, and hence recoverable
in some automatic way.  This is not lost on mathematicians, and many
theorems in the literature~\cite{GPS, mac2013categories} show that certain
higher-categorical structures can be ``strictified''.  The search for
presentations of higher-categorical structures in which certain
operations hold strictly, but which retain the same expressive power,
goes by the name \emph{semi-strict} higher category theory.

\paragraph{Our Contribution}  In this article, we apply
type-theoretic ideas to reduce the complexity of the tower of
coherence data in the definition of higher-dimensional category.  Our
approach is to equip the type theory $\Catt$ \cite{Finster2017} with a
non-trivial definitional equality, written ``$=$'', which strictifies
the part of the theory that handles composition with units, and its
associated higher coherence data.  We call this new theory
$\Catt_\su$, standing for ``\Catt with strict units''.

Since the terms of $\Catt$ describe
higher-categorical operations, this definitional equality yields a new theory where families of  operations coincide.
For example, in a fully weak higher category, the equation
\[ f \circ \id_x = f \]
would need to be witnessed by explicit coherence data. But in $\Catt_\su$ these terms are definitionally equal, and so this unit law holds on-the-nose, and the associated tower of higher-dimensional coherence data trivialises.

We give a detailed analysis of definitional equality in $\Catt_\su$, via a reduction strategy that reduces the complexity of terms. We show that this terminates after finitely many steps, and hence yields for every term $t$ a unique normal form $N(t)$. We then show that this generates definitional equality, in the sense that $t=t'$ if and only if $N(t) \equiv N(t')$, where ``$\equiv$'' represents syntactic equality up to $\alpha$\-equivalence. The existence of distinct normal forms  also shows our theory has non-trivial models.

Definitional equality is hence decidable, and the type theory can be implemented. Our OCaml implementation is made available online at the following address:
\[\text{\url{http://github.com/ericfinster/catt.io/tree/v0.1}}\]

Our results yield the first type-theoretic definition of strictly unital $\infty$\-categories. Previous work has suggested that in such a theory, any pair of valid terms over a disc context should be equal. We prove this property directly for our definition, giving a useful ``sanity check''.

Since our theory $\Catt_\su$ is obtained by adding a definitional equality relation to $\Catt$, there is a trivial mapping $K:\tm(\Catt) \to \tm(\Catt_\su)$.  That is, every $\Catt$ term is automatically a $\Catt_\su$ term. This raises the important question: for every $\Catt_\su$ term $t$, can we find a $\Catt$ term $t'$ such that  $t= K(t')$? We show that this holds for terms over pasting diagrams, thus verifying a form of \textit{conservativity}: every operation in the semi-strict theory corresponds to one in the original weak theory in this case. We apply this result to show that for a weak $\infty$\-category, having strict units is a \emph{property}, rather than extra structure.

Finally, we use our  implementation to construct two substantial examples of terms in our theory, inspired by important algebraic structures in homotopy theory.

\vspace{3pt}
\noindent
-- The \emph{Eckmann-Hilton move} or \emph{braiding}, is a 3-dimensional term that has been previously constructed in \Catt~\cite{Finster2017}. We apply our reduction relation to find its $\Catt_\su$ normal form, obtaining a proof object which is over 50 times smaller than the original.

\vspace{3pt}
\noindent
-- The second example, called the \emph{syllepsis}, is a 5-dimensional algebraic object that had not been explicitly constructed in a pure language of path types prior to this work being advertised.\footnote{One could say that it had been indirectly constructed in homotopy type theory, through a calculation of \(\pi_4(S^3)\)~\cite{brunerie2016homotopy}. Since our work has been completed and advertised, pure path models of the syllepsis have now been constructed by some members of the homotopy type theory community~\cite{Syllepsis}.} We give a formalization of the syllepsis in $\Catt_\su$. Due to its complexity, we have not been able to construct the syllepsis in $\Catt$.

\vspace{3pt}
\noindent
This demonstrates our main research goal: to simplify the syntax of higher category theory so that proofs become smaller, and proof construction becomes easier.

\subsection{Related work}
\label{sec:relatedwork}

\paragraph{Homotopy Type Theory}While our work is not directly concerned with Homotopy Type Theory (HoTT)
\cite{hottbook}, it is nonetheless heavily inspired by developments that these ideas have provoked in the type theory
community.  Indeed, the definition of $\infty$-category presented in
\cite{Finster2017}, and developed in the present  work, was based on a similar
definition of $\infty$\-groupoid~\cite{brunerie2016homotopy}, which in turn can be seen as a distillation of exactly
that part of Martin-L\"{o}f's identity elimination principle which
causes types to behave as higher-dimensional groupoids.

We emphasize that our theory is relevant for directed higher
categories, while HoTT is currently only able to reason about higher
groupoids. At the same time, HoTT is a rich logical system with sums,
products and other higher type formers, while our theory is more
restricted, providing only the type formers for building categorical
operations.

\paragraph{Higher Category Theory}
Our work extends the type theory $\Catt$ of Finster and Mimram, presented at LICS~2017~\cite{Finster2017}. Its models yield an algebraic notion of weak $\infty$-category which agrees with a definition due to Maltsiniotis~\cite{Maltsiniotis}, which itself is a close cousin of Batanin's~\cite{Batanin98} and Leinster's~\cite{leinster} definitions; see~\cite{Ara} for a direct comparison. 

A previous definition of strictly unital weak $\infty$-category has been given by Batanin, Cisinski and Weber~\cite{Batanin2013}, using the mathematical language of operads. That theory has two main axiom classes, \emph{disk reduction} and \emph{unit compatibility}, and we can compare these to our three main generators for definitional equality, \emph{prune}, \emph{disk} and \emph{endo}. While the precise relationship between the theories is not completely clear, we make the following observations.

\vspace{3pt}
\noindent
-- Their \textit{unit compatibility} axiom is close to our \emph{prune} generator, a key point of similarity between the theories.

\vspace{3pt}
\noindent
-- Their \emph{disk reduction} axiom simply requires that terms trivialize over disks, while  for us the corresponding statement is a theorem (see  Theorem~\ref{thm:disctrivialization}.) That is, an {axiom} of BCW  becomes an \textit{emergent property} in our setting.

\vspace{3pt} \noindent -- Our \emph{endo} axiom appears to have no parallel in BCW's theory, and we are able to exhibit a pair of terms which are definitionally equal here, but we believe would \textit{not} be identified in their theory (see page~\pageref{ex:bcw} here.) In this sense, we claim that our theory is likely to be stricter.

\vspace{3pt}

In general, our approach emphasizes an explicit syntax for the
operations of our theory which we are able to manipulate on a
computer, while the BCW approach might be described as more
``semantic'', characterizing the theory by universal properties, but
without an explicit description of the resulting operations.

\section{The Type Theories \texorpdfstring{$\Catt$}{Catt} and \texorpdfstring{$\Catt_{\su}$}{Cattsu}}
\label{sec:typetheory}

It will be convenient to construct our theory in three layers.  We
begin with the raw syntax and basic rules for contexts, types and
substitutions, leaving out the term forming rules.  Most of this
material is standard, and we make sure to point out any idiosyncracies
of notation as we proceed.  After introducing the notion of
\emph{pasting context} we then present the term forming structure of
$\Catt$.  Finally, we introduce some combinatorial material necessary
to describe our equality relation on terms, culminating in the
definition of the theory $\Catt_{\su}$.

\subsection{The Base Theory \Catt}

Here we present the type theory \Catt defined by Finster and Mimram~\cite{Finster2017}.

\paragraph{Raw Syntax}We fix an infinite set $V$ of variables, and use lowercase Roman ($x,y,\dots$) and Greek ($\alpha, \beta, \ldots$) letters to refer to its elements.  The raw syntax of $\Catt$ consists of four syntactic
classes: \emph{contexts}, \emph{types}, \emph{terms} and
\emph{substitutions} (denoted $\Ctx$, $\Type$, $\Term$ and $\Sub$,
respectively).  These classes are defined by the rules in Figure
\ref{fig:syntax}. The arrow constructor \(s \to_A t\) for types represents the type of directed paths from \(s\) to \(t\) where \(s\) and \(t\) have type \(A\). This is similar to the equality type \(s =_A t\) in Martin-L\"of Type Theory which can be thought of as the type of undirected paths between \(s\) and \(t\). Observe that both contexts and substitutions
appear in the raw syntax of terms.

We write $\equiv$ for syntactic equality of the
various syntactic classes up to $\alpha$\-equivalence.

\begin{figure}[b!]
  \vspace{-0.5cm}
  \begin{calign}
\nonumber
    \inferrule{ }{\emptyset : \Ctx}
    &
    \inferrule{\Gamma : \Ctx \\ A : \Type}{\Gamma, A : \Ctx}
\\[5pt]\nonumber
    \inferrule{ }{\star : \Type}
    &
    \inferrule{A : \Type \\ s : \Term \\ t : \Term}{s \to_A t : \Type}
\\[5pt]\nonumber
    \inferrule{v : V}{v : \Term}
    &
    \inferrule{\Gamma : \Ctx \\ A : \Type \\ \sigma : \Sub}{\coh(\Gamma : A)[\sigma] : \Term}
\\[5pt]\nonumber
    \inferrule{ }{\langle \rangle : \Sub}
    &
    \inferrule{\sigma : \Sub \\ t : \Term}{\langle \sigma , t \rangle : \Sub}
    \end{calign}
  \caption{Raw syntax}
  \label{fig:syntax}
\end{figure}

\paragraph{Free Variables} The \emph{free variables} of elements of
each syntactic class are defined by induction on the structure as follows, where $x \in V$:%
\begin{align*}
  \FV(\emptyset) &= \emptyset &
  \FV(\Gamma, x : A) &= \FV(\Gamma) \cup \{ x \}\\
  \FV(\star) &= \emptyset &
  \FV(s \to_A t) &= \FV(A) \cup \FV(s) \cup \FV(t) \\
  \FV(x) &= \{ x \} &
  \FV(\coh(\Gamma : A)[\sigma]) &= \FV(\sigma)\\
  \FV(\langle \rangle) &= \emptyset &
  \FV(\langle \sigma , x \mapsto t \rangle) &= \FV(\sigma) \cup \FV(t) \\
       \end{align*}

\vspace{-17pt}
\paragraph{Dimension}We define the \emph{dimension} of a type by
induction:
\begin{align*}
\dim \, \star &= -1 & \dim \, (s \to_A t) &= \dim \, A + 1
\intertext{We extend this notion to contexts by asserting that the dimension of a
context is one more than the maximum of the dimension of the types
occurring in that context.}
  \dim \, \emptyset &= -1 &\dim \, (\Gamma, x : A) &= \max(\dim \, \Gamma , \dim \, A + 1)
\end{align*}

\vspace{-4pt}
\paragraph{Term Substitutions and Compositions}Although the terms of $\Catt$ are always in normal form, we will need
to perform actual substitutions on terms during type-checking.  We
therefore also define a semantic form of substitution which calculates
by induction on the structure of terms.  We denote this operations by
$\sub{-}$ in order to distinguish if from the $[-]$ appearing in
coherence terms, which is part of the syntactic structure.  This
operation is defined as follows:
\begin{align*}
  \star\sub{\sigma} &= \star \\*
  (s \to_{A} t)\sub{\sigma} &=  s\sub{\sigma} \to_{A\sub{\sigma}} t\sub{\sigma} \\*
  \coh (\Gamma : s \to t)[\tau]\sub{\sigma} &= \coh (\Gamma : s \to t)[\tau \circ \sigma] \\*
  x\sub{\sigma} &= t \qquad \text{if $x \mapsto t \in \sigma$}
       \end{align*}
Composition of substitutions, written ``$\circ$'', is defined mutually recursively as follows:
\begin{align*}
  \langle \rangle \circ \sigma &= \langle \rangle
  &
  \langle \tau , t \rangle \circ \sigma &= \langle \tau \circ \sigma , t\sub{\sigma} \rangle
\end{align*}
For every context \(\Gamma\) we also have an identity substitution \(\id_\Gamma\), mapping each variable to itself. By simple induction we get the following proposition.

\begin{restatable}{proposition}{cattiscat}\label{prop:catt-is-cat}
  Contexts and substitutions form a category; that is, the following syntactic equations hold, for context \(\Gamma\), type \(A\) and term \(s\) of \(\Gamma\), and composable substitutions \(\mu, \sigma, \tau\):
  \begin{align*}
    A\sub{\sigma \circ \tau} &\equiv A\sub\sigma\sub\tau
    &A\sub{\id_\Gamma} &\equiv A\\
    s\sub{\sigma \circ \tau} &\equiv s\sub\sigma\sub\tau
    &s\sub{\id_\Gamma} &\equiv s\\
    \mu \circ (\sigma \circ \tau) &\equiv (\mu \circ \sigma) \circ \tau & \id_\Delta \circ \mu &\equiv \mu \equiv \mu \circ \id_\Gamma
  \end{align*}
\end{restatable}
\begin{proof}
  This is a simple mutual induction over the syntax.
\end{proof}

The basic typing judgements for
contexts, types and substitutions are given in Figure~\ref{fig:basic-rules}.  The rules are standard for a dependent type
theory.  Note that our types consist of just a single base type
denoted $\star$, and a formation rule analogous to the formation rule
for identity types in Martin L\"{o}f Type Theory. The fact that this
rule captures faithfully the notion of globular set is at the heart of
the connection between type theory and higher category theory, and is
the basis of this syntactic description of $\infty$-categories.

\begin{figure}[b!]
  \vspace{-.5cm}
\begin{calign}
\nonumber
    \inferrule{ }{\emptyset \vdash}
    &
    \inferrule{\Gamma \vdash \\ \Gamma \vdash A : \Type}{\Gamma, x : A \vdash}
    \\[5pt]
\nonumber
    \inferrule{\Gamma \vdash}{\Gamma \vdash \star : \Type}
    &
    \inferrule{\Gamma \vdash A : \Type \\ \Gamma \vdash a : A \\ \Gamma \vdash b : A}{
      \Gamma \vdash a \to_A b : \Type
    }
\\[5pt]
\nonumber
    \inferrule{\Gamma \vdash}{\Gamma \vdash \langle \rangle : \emptyset}
    &
    \inferrule{
      \Gamma \vdash \sigma : \Delta \\
      \Delta \vdash A : \Type \\
      \Gamma \vdash t : A[\sigma]
    }{
      \Gamma \vdash \langle \sigma , x \mapsto t \rangle : \Delta, x : A
    }
\end{calign}

\vspace{-8pt}
  \caption{Basic typing rules}
  \label{fig:basic-rules}
\end{figure}

\paragraph{Support} Given a term \(t\) and context \(\Gamma\), we can define the \emph{support} of \(t\) in \(\Gamma\) to be the downwards closure of \(FV(t)\) in \(\Gamma\), where a set of variables \(S\) is downwards closed if for all \(x \in S\), \(x : A \in \Gamma\), and \(y \in \FV(A)\) we have \(y \in S\). If the context is obvious we simply write \(\supp(t)\) for the support of \(t\).

\paragraph{Pasting Contexts}The terms of $\Catt$ are derived from isolating a distinguished subset
of contexts which we call \emph{pasting contexts}.  A set of rules for
exhibiting evidence that a given context is a pasting context was the
key innovation of~\cite{Finster2017}. These rules are presented in
Figure~\ref{fig:ps-rules}.

\begin{figure}[b!]
  \vspace{-.5cm}
\begin{calign}
\nonumber
    \inferrule*[Right=$\star$]{ }{x : \star \vdash_{p} x : \star}
    &
    \inferrule*[Right=$\checkmark$]{\Gamma \vdash_{p} x : \star}{\Gamma \vdash_{p}}
\\[5pt]\nonumber
    \inferrule*[Right=$\Uparrow$]{\Gamma \vdash_{p} x : A}{
      \Gamma, y : A, f : x \to_A y \vdash_{p} f : x \to_A y
    }
    &
    \inferrule*[Right=$\Downarrow$]{\Gamma \vdash_{p} f : x \to_A y}{\Gamma \vdash_{p} y : A}
\end{calign}

\vspace{-8pt}
  \caption{Pasting contexts}
  \label{fig:ps-rules}
\end{figure}

\paragraph{Boundary Variables}For each pasting context $\Gamma \vdash_{p}$, we will define two
distinguished subsets of the variables, denoted $\partial^-(\Gamma)$
and $\partial^+(\Gamma)$. First, for a variable $x : A \in \Gamma$,
define its dimension to be $\dim \, A + 1$.  Furthermore,
let us say that a variable is \emph{target-free} if it does \emph{not}
occur as the target of any other variables in $\Gamma$.  Similarly,
we have the notion of \emph{source-free}.  We now define:
\begin{alignat*}{2}
  \partial^-(\Gamma) &:= \{ x \in \Gamma \, | \, &&\dim \, x < \dim \, \Gamma - 1,\\
  &&&\textit{or, } \dim \, x = \dim \, \Gamma - 1 \, \text{and $x$ is target-free} \}\\
  \partial^+(\Gamma) &:= \{ x \in \Gamma \, | \, &&\dim \, x < \dim \, \Gamma - 1,\\
  &&&\textit{or, } \dim \, x = \dim \, \Gamma - 1 \, \text{and $x$ is source-free} \}
\end{alignat*}

\paragraph{Terms}With these notions in place, the typing rules for
terms of $\Catt$ are shown in Figure~\ref{fig:terms}. Note that when
writing the substitution in a coherence term, we typically omit the
angled brackets, writing $\coh(\Gamma : s \to_A t)[a, b, c, \dots]$
instead of
$\coh(\Gamma : s \to_A t)[\langle a, b , c, \dots \rangle]$.

\begin{figure}[b!]
    \vspace{-.5cm}
  \begin{mathpar}
    \inferrule{\Gamma \vdash \\ (x : A) \in \Gamma}{\Gamma \vdash x : A} \and
    \inferrule{
      \Gamma \vdash_{p} \\
      \Gamma \vdash s \to_A t \\
      \Delta \vdash \sigma : \Gamma \\
      \supp(s) = \partial^-(\Gamma) \\
      \supp(t) = \partial^+(\Gamma)
    }{
      \Delta \vdash \coh(\Gamma : s \to_A t)[\sigma] : s\sub{\sigma} \to_{A\sub{\sigma}} t\sub{\sigma}
    } \and
    \inferrule{
      \Gamma \vdash_p \\
      \Gamma \vdash s \to_A t \\
      \Delta \vdash \sigma : \Gamma \\
      \supp(s) = \supp(t) = \FV(\Gamma)
    }{
      \Delta \vdash \coh(\Gamma : s \to_A t)[\sigma] : s\sub{\sigma} \to_{A\sub{\sigma}} t\sub{\sigma}
    }
  \end{mathpar}

\vspace{-8pt}
  \caption{Term constructors}
  \label{fig:terms}
\end{figure}
\paragraph{Examples}We record here some basic examples of well-typed
terms.

\begin{align*}
  \mathsf{comp_1} &:= \coh((x : \star)(y : \star)(f : x \to_\star y): x \to_\star y)[x,y,f]
\\[3pt]
\mathsf{comp_2} &:= \coh((x : \star)(y : \star)(f : x \to_\star y) (z : \star) (g : y \to_\star z)
\\[-2pt]
&\qquad : x \to_\star z)[x,y,f,z,g]
\\[3pt]
\mathsf{comp_3} &:= \coh((x : \star)(y : \star)(f : x \to_\star y) (z : \star) (g : y \to_\star z)
\\[-2pt]
&\qquad (w : \star) (h : z \to_\star w) : x \to_\star w)[x,y,f,z,g,w,h]
\\[3pt]
\i_0 &:= \coh((x : \star) : x \to_{\star} x)[x]
\\[3pt]
\mathsf{unit\mhyphen r} &:= \coh((x : \star)(y : \star)(f : x \to_\star y)\\[-2pt]
  &\qquad:\mathsf{comp_2}\sub{x,y,f,y,\i_0\sub{y}} \to_{x \to_\star y} f)[x,y,f]
  \\[5pt]
  \mathsf{comp_{2,0}} &:= \coh((x : \star)(y : \star)(f : x \to_\star y) (z : \star)(g : y \to_\star z)
\\[-2pt]
&\qquad(h : y \to_\star z)(\alpha : g \to_{y \to_\star z} h) (w : \star)(k : z \to_\star w)
\\[-2pt]
& \qquad : \mathsf{comp_3}\sub{x,y,f,z,g,w,k} \to_{(x \to_\star w)}
\\[-2pt]
&\qquad\hspace{7pt}\mathsf{comp_3}\sub{x,y,f,z,h,w,k})[x,y,f,z,g,h,\alpha,w,k]
\end{align*}
In the theory of $\infty$\-categories, these correspond to  \textit{operations}  from  low-dimensional higher category theory, such as the theory of bicategories. In that more familiar language, given objects $x,y,z,w$ and 1\-morphisms $f:x \to y$, $g :y \to z$, $h:z \to w$, we have the following, where we write $\cdot$ for forward composition of 1\-morphisms:
\begin{itemize}
\item $\comp_1 \sub{x,y,f}$ gives the unary composite $(f)$;
\item $\comp_2 \sub{ x, y, f, z, g}$ gives the binary composite $f \cdot g$;
\item $\comp_3 \sub{x, y, f, z, g, w, h}$ gives the unbiased ternary composite \mbox{$f \cdot g \cdot h$};
\item $\comp_2 \sub{x, z, \comp_2 \sub{x, y, f, z, g}, w, h}$ gives the iterated binary composite $(f \cdot g) \cdot h$;
\item $\i_0 \sub x$ corresponds to the 1\-morphism $\id_x$;
\item $\mathsf{unit\mhyphen r} \sub f$ gives to the unit 2\-morphism $f \cdot \id_y \Rightarrow f$.
\end{itemize}
The unary composite $(f)$ and the unbiased ternary composite $f \cdot g \cdot h$ are not directly defined in the traditional notion of bicategory. We could similarly write down operations for other familiar operations in the theory of weak $\infty$\-categories, such as associators, interchangers, and so on, in principle in all dimensions. It is in this sense that \Catt gives a formal language for weak $\infty$\-categories.

This last example $\comp_{2,0}$ is part of a family of coherences
$\mathsf{comp}_{d,k}$ for $k < d$, \label{page:comp}which will play a role in Section~\ref{sec:rehydration}.  While we will not give a formal definition,
these compositions can be described intuitively as follows: they
consist of the ``unbiased'' composite of a $d$-dimensional disc $D$
with two $(k+1)$-dimensional discs $S$ and $T$ glued to the
$k$-dimensional source and target of $D$, respectively.  In traditional notation we might write this as $S \circ_{k} D \circ_{k} T$.

\paragraph{Identity Terms} Generalizing the 0-dimensional case above,
we can define an identity on cells of arbitrary dimension.  To do
so, we assume that the set $V$ of variables contains elements
$d_i$ and $d_i'$ for $i \in \N$.  Now define the \emph{$k$-disc
  context} and the \emph{$(k-1)$-sphere type} by mutual induction on
$k$ as follows:
{\small\begin{align*}
  \uD^{0} &:= \emptyset, (d_0 : \star) & S^{-1} &:= \star \\
  \uD^{k+1} &:= \uD^{k}, (d_k' : S^{k-1}) , (d_{k+1} : S^k) & S^{k} &:= d_k \to_{S^{k-1}} d_k'
\end{align*}}%
Finally, for $k \in \N$, we define the \textit{identity on the $k$\-disc},
written $\i_k$ as follows, a valid term in the context $\uD^k$:
\[ \i_k := \coh ( \uD^k: d_k \to_{S^{k-1}} d_k)[\id _{D^k}] \]

\paragraph{Substitutions From a Disc} A substitution out of a \(k\)-disc context contains the same data as a type and a term. Given a type $A$ and term $t$, we define a substitution $\{A,t\}$ as follows:
\begin{align*}
\{\star,x\} &:= \langle x \rangle
&
\{ u \to_A v, x\} &:= \langle \{A, u\}, v, x \rangle
\end{align*}
Moreover, given a substitution \(\sigma\) from a disc \(D^k\), we have \(\sigma \equiv \{S^{k-1}\sub\sigma, d_k\sub\sigma\}\). We also have the following lemma describing the interaction between substitutions from a disc and composition.

\begin{restatable}{lemma}{typeunfoldingsub}
\label{lem:typeunfoldingsub}
Substitutions from a disc compose in the following way:
$$\{A \sub \sigma, t \sub \sigma \} \equiv \{A , t\} \circ \sigma$$
\end{restatable}
\begin{proof}
If $A \equiv \star$ then the lemma reduces to \(\langle t \rangle \equiv \langle t \rangle\). If $A \equiv u \to_{A'} v$, then we reason by induction on subtypes as follows:
\begin{align*}
\{ A \sub \sigma, t \sub \sigma \} &\equiv \{u \sub \sigma \to_{A' \sub \sigma} v \sub \sigma, t \sub \sigma \}\\
                                   &\equiv \langle \{A' \sub \sigma, u \sub \sigma\}, v \sub \sigma, t \sub \sigma \rangle\\
                                   &\equiv \langle \{A',u\} \sub \sigma, v \sub \sigma, t \sub \sigma \rangle\\
                                   &\equiv \langle \{A',u\}, v, t \rangle \circ \sigma\\
                                   &\equiv \{u \to_{A'} v, t\} \circ \sigma\\
                                   &\equiv \{A,t\} \circ \sigma
\end{align*}
This completes the proof.
\end{proof}

\paragraph{Inferred Types} If a term $t$ is well-scoped in some context \(\Gamma\) (i.e. all variables it contains are either bound or in \(\Gamma\)), we can derive a canonical type for it $\ty(t)$, which we call the \emph{inferred type}. For a variable \(x\), by assumption there is \((x,A) \in \Gamma\) for some type \(A\), and we define $\ty(x) := A$. For a coherence term, we define the inferred type as follows:
\[\mathsf{ty}(\coh(\Gamma : U)[\sigma]) := U \sub\sigma\]
Furthermore, when the inferred type of a term \(u\) is of the form \(s \to_A t\), we define the \textit{inferred source}, \(\src(u)\), to be \(s\) and the \textit{inferred target} \(\tgt(u)\) to be \(t\). We note that a term of the form $\coh(\Gamma : \star)[\ldots]$ is never valid, so $\src,\tgt$ are defined for every valid coherence term.

We  write $\src^k, \tgt^k$ for the iterated $k$\-fold inferred source or target, and $\src_k$, $\tgt_k$ for the $k$\-dimensional inferred source or target; so for a coherence term $t$ of dimension $n$, we have $\src_k(t) := \src^{n-k}(t)$, and $\tgt_k(t) := \tgt^{n-k}(t)$.

Using inferred types we can extend the notion of dimension to terms, by defining $\dim \, t := \dim \,\ty(t) + 1$.

\vspace{-5pt}
\subsection{The type theory \texorpdfstring{$\Catt_\su$}{Cattsu}}
\label{sec:cattsu}

As described in the introduction, the type theory $\Catt$ of the
previous section contains no non-trivial definitional equalities:
while calculation happens during type-checking, all terms themselves
are in normal form.  In this section, we introduce our equality
relation.  Its definition will require some combinatorial preparation
which we turn to now.
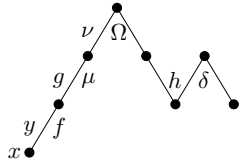
\begin{figure}[b!]\vspace{-.5cm}
  \[
    \tikzset{point/.style={draw, circle, black, fill, inner sep=1.5pt, label={}}}
    \begin{aligned}
      \scalebox{0.8}{\begin{tikzpicture}[xscale=.6, scale=.8]
        \draw [use as bounding box, draw=none] (0,0) rectangle +(8,3);
        \draw (0,0) node [point] {} node [left] {$x$} to node [left] {$y$} node [right] {$f$} (1,1) node [point] {} to node [left] {$g$} node [right] {$\mu$} (2,2) node [point] {} to node [left] {$\nu$} node [right] {$\Omega$} (3,3) node [point] {} to (4,2) node [point] {} to (5,1) node [point] {} to node [left] {$h$} node [right] {$\delta$} (6,2) node [point] {} to (7,1) node [point] {};
      \end{tikzpicture}}
    \end{aligned}
  \]

\vspace{-8pt}
  \caption{An element of
  \\
  $\Dyck \, 1 \Downarrow (\Uparrow (\Downarrow (\Downarrow (\Uparrow (\Uparrow (\Uparrow (\star x) \, y \, f) \, g \, \mu) \, \nu \, \Omega))) \, h \, \delta)$}
  \label{fig:dyck-ex}
\end{figure}

\paragraph{Dyck Words}Observe that each of the rules for pasting contexts in Figure
\ref{fig:ps-rules} has at most one hypothesis, and consequently,
derivations made with these rules are necessarily linear.  In fact,
complete derivations of the fact that a context is a pasting context
can be identified with \emph{Dyck words}~\cite{kozen2012automata}.
Strictly speaking, a pasting context is defined as a pair of a
syntactic context and a derivation of the fact that it is well-formed,
but this representation contains redundancy and can
be awkward to manipulate in practice.  For
presentation, it is more convenient to work with \emph{Dyck words}, which are a somewhat
simplified representation of pasting diagrams.

Dyck words may be pictured as a list of up and down moves, with
each up move labeled by a pair of variable names.  Concretely, the set $\Dyck \, n$ of
\emph{Dyck words of excess~$n$} is defined by the rules given
in Figure~\ref{fig:dyck}.  The parameter~$n$ records
the difference between the number of up and down moves.  The definition ensures that the excess is always non-negative, so
we always have at least as many up moves as down moves, leading
to the ``mountain'' diagram of Figure~\ref{fig:dyck-ex}.

The rules
for Dyck words mirror exactly the derivation rules for pasting
contexts of Figure~\ref{fig:ps-rules}.  (The additional $\checkmark$
rule for pasting contexts forces a complete derivation
$\Gamma \vdash_p$ to be of excess 0).  As an example, consider this context:
\[ (x : *) (y : *) (f : x \to_\star y) (z : *) (g : y \to_\star z) \]
It is proven to be a pasting context via this derivation:
\vspace{4pt}
\begin{mathpar}
  \inferrule*[Right=$\Downarrow$]{
    \inferrule*[Right=$\Uparrow$]{
      \inferrule*[Right=$\Downarrow$]{
        \inferrule*[Right=$\Uparrow$]{
          \inferrule*[Right=$\star$]{ }{(x : \star) \vdash_p (x : \star)}
        }{(x : \star)(y : \star)(f : x \to_\star y) \vdash_p (f : x \to_\star y)}
      }{(x : \star)(y : \star)(f : x \to_\star y) \vdash_p (y : \star)}
    }{(x : \star)(y : \star)(f : x \to_\star y)(x : \star)(g : y \to_\star z) \vdash_p (g : y \to_\star z)}
  }{(x : \star)(y : \star)(f : x \to_\star y)(x : \star)(g : y \to_\star z) \vdash_p (z : \star)}
\end{mathpar}

Its Dyck word representation is the following:

\[ \Downarrow \, (\Uparrow \, (\Downarrow \, (\Uparrow \, (\star \, x)
  \, y \, f) ) \,z \, g) ) : \Dyck \, 0 \] The only
difference between these two representations is that the pasting
context, together with its derivation, remembers all the typing
information of the all variables, while the Dyck word representation
remembers just the variable names.  Since the types can be recovered
from the structure of the Dyck word itself, this represents no loss of
information. To fix notation, if $\Gamma \vdash_p (x : A)$ is a
pasting context, we write $\lfloor \Gamma \rfloor : \Dyck \, (\dim \, A)$ for
the corresponding Dyck word. Conversely, for a Dyck
word \mbox{$d : \Dyck \, n$}, we write $\lceil d \rceil$, \(\ty(d)\), and \(\tm(d)\) for the corresponding context, type and term such that \(\lceil d \rceil \vdash_p \tm(d) : \ty(d)\).

\begin{figure}[th!]
  \begin{mathpar}
    \inferrule{x \in V}{\star \, x : \Dyck \, 0} \and
    \inferrule{n : \N \\ d : \Dyck \, n \\ y , f : V}{
      \Uparrow \, d \, y \, f : \Dyck \, (S \, n)
    } \and
    \inferrule{n : \N \\ d : \Dyck \, (S \, n)}{
      \Downarrow \, d : \Dyck \, n
    }
  \end{mathpar}

\vspace{-8pt}
\caption{Dyck words}
  \label{fig:dyck}
\end{figure}
\begin{figure}[b!]
\def\littlegap{\hspace{10pt}}
  \begin{mathpar}
    \inferrule{n : \N \littlegap d : \Dyck \, n \littlegap y , f : V}{
      \Updownarrow_{\mathsf{pk}} \, d : \Peak \, (\Downarrow (\Uparrow \, d \, y \, f))
    } \,\,\,\,\,\,\,\,\,\,
    \inferrule{n : \N \littlegap d : \Dyck \, n \hspace{7pt} y , f : V \littlegap p : \Peak \, d}{
      \Uparrow_{\mathsf{pk}} \, d \, y \, f \, p : \Peak \, (\Uparrow \, d \, y \, f)
    } \and
    \inferrule{n : \N \littlegap d : \Dyck \, (S \, n) \littlegap p : \Peak \, d}{
      \Downarrow_{\mathsf{pk}} \, d \, p : \Peak \, (\Downarrow \, d)
    }
  \end{mathpar}

\vspace{-8pt}
  \caption{Peaks}
  \label{fig:peaks}
\end{figure}

\paragraph{Peaks}A special role will be played by the
positions in a Dyck word where we change direction from moving
up to moving down.  We call these the \emph{peaks}.  It is easy
to give an induction characterization of peaks, and we do so in
Figure~\ref{fig:peaks}. A nice advantage of this representation is that we can write programs
on pasting contexts using pattern matching style (for which we
use an Agda-style syntax).

If $\Gamma \vdash_p$ is a pasting context, we say that a variable which occurs as the label of a
peak in the Dyck word representation of $\Gamma$ is \emph{locally
  maximal}.  Intuitively speaking, such variables represent those
cells which are of highest dimension ``in their neighborhood.''  We
write $\LM(\Gamma)$ for the set of locally maximal variables and for
$\alpha \in \LM(\Gamma)$ we write $p_\alpha$ for the corresponding
peak.  Observe that if $\alpha \in \LM(\Gamma)$, we may speak
of the \emph{dimension} $\dim \, \alpha\,$ since $\alpha$ is assigned
a type by $\Gamma$.  Moreover, since all locally maximal variables are
positive dimensional, the inferred type of any locally maximal variable
$\alpha \in \LM(\Gamma)$ must be an arrow type
$\alpha : s \to_A t$ for some type $A$. Therefore the \(\src\) and \(\tgt\) of a locally maximal are always well defined. We note that any pasting context \(\Gamma\) contains no coherence terms, and so \(\src(\alpha)\) and \(\tgt(\alpha)\) are always variables.

We can use the definition of peaks to manipulate our Dyck word in
various ways.  We define three operations in Figure~\ref{fig:operations} which will be useful in what follows.

\begin{figure}\small
\begin{align*}
  &\mathsf{excise} : \{n : \N\} \{d : \Dyck \, n\} (p : \Peak \, d) \to \Dyck \, n\\[-1pt]
  &\mathsf{excise} \, (\Updownarrow_{pk} \, d \, y \, f) = d \\[-1pt]
  &\mathsf{excise} \, (\Uparrow_{pk} \, d \, y \, f \, p) = \, \Uparrow (\mathsf{excise} \, p) \, y \, f \\[-1pt]
  &\mathsf{excise} \, (\Downarrow_{pk} \, d \, p) = \,\Downarrow (\mathsf{excise} \, p)\\[5pt]
  &\mathsf{project} : \{n : \N\} \{d : \Dyck \, n\} (p : \Peak \, d) \to \Sub\\[-1pt]
  &\mathsf{project} \, (\Updownarrow_{pk} \, d \, y \, f) = \langle \id_{\lceil d \rceil}, \tm(d), \i(\ty(d),\tm(d)) \rangle\\[-1pt]
  &\mathsf{project} \, (\Uparrow_{pk} \, d \, y \, f \, p) = \langle \mathsf{project} \, p , y , f \rangle\\[-1pt]
  &\mathsf{project} \, (\Downarrow_{pk} \, d \, p) = \mathsf{project} \, p\\[5pt]
  &\mathsf{remove} : \{n : \N\} \{d : \Dyck \, n\} (p : \Peak \, d) (\sigma : \Sub) \to \Sub\\[-1pt]
  &\mathsf{remove} \, (\Updownarrow_{pk} \, d \, y \, f) \, \langle \sigma , s , t \rangle = \sigma\\[-1pt]
  &\mathsf{remove} \, (\Uparrow_{pk} \, d \, y \, f \, p) \, \langle \sigma , s , t \rangle = \langle \mathsf{remove} \, p \, \sigma , s , t \rangle\\[-1pt]
  &\mathsf{remove} \, (\Downarrow_{pk} \, d \, p) \, \sigma = \mathsf{remove} \, p \, \sigma
\end{align*}

\vspace{-8pt}
  \caption{Operations on Dyck words}
  \label{fig:operations}
\end{figure}

The first operation is \emph{excising} a peak of a Dyck word. Given a pasting context $\Gamma \vdash_p$ and a locally maximal
variable $\alpha : \LM(\Gamma)$, we write $\Gamma\sslash \alpha$ for
the pasting context obtained by excising the peak where $\alpha$
occurs in the Dyck word representation of~$\Gamma$:
\[\Gamma\sslash \alpha := \lceil \mathsf{excise} \, p_{\alpha} \rceil\]
This operation removes \(\alpha\) and \(\tgt(\alpha)\) from the original context. A pictorial representation is given in Figure~\ref{fig:excision}.

The second operation, \emph{project}, is used to create a natural substitution $\Gamma\sslash \alpha \vdash \pi_{\alpha} : \Gamma$, given by:
\[\pi_\alpha := \mathsf{project}\,p_\alpha\]
If \(\ty(\alpha) \equiv \src(\alpha) \to_A \tgt(\alpha)\) then this sends the variable \(\alpha\) to \(\i_{\dim \, A}\sub{\{A,\src(\alpha)\}}\), the variable \(\tgt(\alpha)\) to \(\src(\alpha)\), and every other variable to itself.

Lastly we also define an operation, \emph{remove}, that removes two terms from a given substitution. If $\Delta$ is another context, and we
are given a substitution $\Delta \vdash \sigma : \Gamma$, then we have
a substitution
$\Delta \vdash \sigma\sslash \alpha : \Gamma\sslash \alpha$ obtained by the remove function:
\[\sigma \sslash \alpha := \mathsf{remove}\, p_\alpha\, \sigma\]
which removes the terms in \(\sigma\) corresponding to \(\alpha\) and \(\tgt(\alpha)\).

We will show that these are actually well typed constructions later in Lemma~\ref{lem:pruning-construct}.

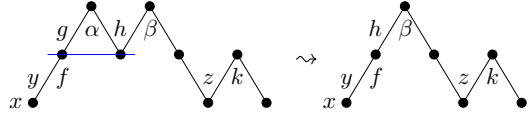
\begin{figure}[!h]\vspace{-.5cm}
  \[
    \tikzset{point/.style={draw, circle, black, fill, inner sep=1.5pt, label={}}}
    \begin{aligned}
      \scalebox{0.8}{\begin{tikzpicture}[xscale=.6, scale=.8]
        \draw [use as bounding box, draw=none] (0,0) rectangle +(8,3);
        \draw (0,0) node [point] {} node [left=1pt] {$x$} to node [left] {$y\vphantom{yk}$} node [right] {$f\vphantom{yk}$} (1,1) node [point] {} to node [left] {$g\vphantom{yk}$} node [right] {$\alpha\vphantom{yk}$} (2,2) node [point] {} to (3,1) node [point] {} to node [left] {$h\vphantom{yk}$} node [right] {$\beta\vphantom{yk}$} (4,2) node [point] {} to (5,1) node [point] {} to (6,0) node [point] {} to node [left] {$z\vphantom{yk}$} node[right] {$k\vphantom{yk}$} (7,1) node [point] {} to (8,0) node [point] {};
        \draw [blue] (0.5,1) to +(3,0);
      \end{tikzpicture}}
    \end{aligned}
    \quad\raisebox{-8pt}{$\squig$}\quad
    \begin{aligned}
      \scalebox{0.8}{\begin{tikzpicture}[xscale=.6, scale=.8]
        \draw [use as bounding box, draw=none] (0,0) rectangle +(6,3);
        \draw (0,0) node [point] {} node [left=1pt] {$x$} to node [left] {$y\vphantom{yk}$} node [right] {$f\vphantom{yk}$} (1,1) node [point] {} to node [left] {$h\vphantom{yk}$} node [right] {$\beta\vphantom{yk}$} (2,2) node [point] {} to (3,1) node [point] {} to (4,0) node [point] {} to node [left] {$z\vphantom{yk}$} node[right] {$k\vphantom{yk}$} (5,1) node [point] {} to (6,0) node [point] {};
      \end{tikzpicture}}
    \end{aligned}
  \]

\vspace{-8pt}
  \caption{Excising a peak}
  \label{fig:excision}
\end{figure}

\paragraph{Equality in $\Catt_{\su}$} With these definitions in place,
we can now define our equality relation on terms in Figure~\ref{fig:eq-rules}. Our equality is indexed by the context which we are working with, and we will write statements like \(\Gamma \vdash s = t\) to mean terms \(s\) and \(t\) are equal in context \(\Gamma\). Where the context is clear we may drop it and simply write $s = t$.
There are three generating equalities, which we summarize as follows.

\begin{figure}[b!]
\vspace{-.4cm}

  \begin{mathpar}
    \inferrule*[Right=prune]{
      \Gamma \vdash \coh(\Delta : A)[\sigma] : B \\ \alpha : \LM(\Delta) \\ \alpha\sub{\sigma} \equiv \i_{\dim \, \alpha - 1}\sub{\tau}
    }{
      \Gamma \vdash \coh(\Delta : A)[\sigma] = \coh(\Delta\sslash \alpha : A\sub{\pi_{\alpha}})[\sigma\sslash x]
    } \and
    \inferrule*[Right=disc]{
      n : \N \\ \Gamma \vdash \coh(\uD^{n+1} : S^n)[\sigma] : B
    }{
      \Gamma \vdash \coh(\uD^{n+1} : S^n)[\sigma] = d_{n+1}\sub{\sigma}
    } \and
    \inferrule*[Right=endo]{
      \Gamma \vdash \coh (\Delta : t \to_A t)[\sigma] : B
    }{
      \Gamma \vdash \coh (\Delta : t \to_A t)[\sigma] = \i_{\dim \, A}\sub{\{A,t\} \circ \sigma}
    }
  \end{mathpar}

\vspace{-8pt}
  \caption{Generating  equality judgements on terms}
  \label{fig:eq-rules}
\end{figure}

\begin{itemize}
\item \emph{Pruning} scans the locally maximal arguments of a substitution looking for
identity terms.  When such an argument appears, it may be removed,
while at the same time removing the corresponding Dyck peak
from the pasting diagram defining the coherence.
\item \emph{Disc Removal} asserts that unary composites may be removed from
the head of a term.
\item \emph{Endomorphism Coherence Removal}
asserts that coherences associated to a repeated term may be replaced with
identities on that term.
\end{itemize}

\noindent
We extend this equality relation to types and
substitutions by structural induction on the formation rules. It is sufficient to add rules for reflexivity, transitivity, and symmetry to our term equality only; these can then be proven for the other equalities by a simple induction. The remaining rules for our equality are given in  Figure~\ref{fig:structural-eq}.

\begin{figure}[th!]
  \begin{mathpar}
    \inferrule{\Gamma : \Ctx \\ (x,A) \in \Gamma}{\Gamma \vdash x = x}\and
    \inferrule{\Delta \vdash A = B \\ \Gamma \vdash \sigma = \tau}{\Gamma \vdash \coh(\Delta : A)[\sigma] = \coh(\Delta : B)[\tau]}\and
    \inferrule{\Gamma \vdash s = t} {\Gamma \vdash t = s}\and
    \inferrule{\Gamma \vdash s = t \\ \Gamma \vdash t = u}{\Gamma \vdash s = u}\and
    \inferrule{\Gamma : \Ctx}{\Gamma \vdash \star = \star} \and
    \inferrule{\Gamma \vdash A = A' \\ \Gamma \vdash s = s' \\ \Gamma \vdash t = t'}{
      \Gamma \vdash s \to_A t = s' \to_{A'} t'
    }\and
    \inferrule{\Gamma : \Ctx}{\Gamma \vdash \langle  \rangle = \langle  \rangle} \and
    \inferrule{\Gamma \vdash \sigma = \tau\\ \Gamma \vdash s = t}{\Gamma \vdash \langle \sigma , x \mapsto s \rangle = \langle \tau , x \mapsto t \rangle}
  \end{mathpar}

\vspace{-8pt}
  \caption{Structural equality rules}
  \label{fig:structural-eq}
\end{figure}

Finally, to integrate our new equality with the type system, we
require our term judgement to be equipped with a conversion rule, which allows typing to interact with our equality.
This rule is listed in Figure~\ref{fig:conversion}.

\begin{figure}[bh!]
  \begin{mathpar}
    \inferrule{\Gamma \vdash t : A \\ \Gamma \vdash A = B
    }{\Gamma \vdash t : B} \and
  \end{mathpar}

\vspace{-8pt}
  \caption{Conversion rule}
  \label{fig:conversion}
\end{figure}

\paragraph{Example Reductions}We record here examples of our three generating reductions, in order
to give the reader a flavour of how they operate.

\paragraph{$\bullet$ Pruning}If a coherence term has a substitution which sends a locally-maximal variable $\alpha$ to an identity, then the pruning relation allows $\alpha$ to be removed from the corresponding pasting context. Consider the following term:
  \begin{align*}
    &\mathsf{comp_2}\sub{x,y,f,y,\i_0\sub{y}}\\
                   &\hspace{0.5cm}:= \coh((x : \star)(y : \star)(f : x \to_\star y) \\
    &\hspace{1.7cm} (z : \star)(g : y \to_\star z) : x \to_\star z)[x,y,f,y,\i_0\sub{y}]
  \end{align*}
The variable $g$ occurs in a locally maximal position in the head. Moreover, the argument supplied in this position is $\i_0\sub{y}$.  Hence
  the pruning relation applies.  We have:
  \begin{align*}
    \Gamma\sslash g &:= (x : \star)(y : \star)(f : x \to_\star y) \\*
    \pi_g &:= [x,y,f,y,\i_0\sub{y}] \\*
    [x,y,f,y,\i_0\sub{y}]\sslash g &:= [x,y,f]
  \end{align*}
so that this term is definitionally equal to the following:
\[
\coh((x : \star)(y : \star)(f : x \to_\star y) : x \to_\star y)[x,y,f]
\] \label{ex:eq-exs}%
In more standard notation, we have shown $f \cdot \id_x = (f)$, where $(f)$ represents the unary composite of $f$.

\paragraph{$\bullet$ Disc Removal} This relation says that when we are composing over a disc context, with the sphere type, the entire term is definitionally equal to its last argument. This case applies for the unary composite $(f)$ obtained just above. In this case, application of the \textsc{Disc} rule now yields:
  \[ \coh((x : \star)(y : \star)(f : x \to_\star y) : x \to_\star
    y)[x,y,f] = f \] Note that this replaces the coherence term with the final argument, $f$, of its substitution $[x, y, f]$.  As $f$ is now a variable, the term is now in normal form (see Section 3).

\paragraph{$\bullet$ Endomorphism Coherence Removal}A curious
  redundancy of the fully weak definition of $\infty$-category is the
  existence of ``fake identities'': cells which are ``morally'' the
  identity on some composite cell, but do not have an identity
  coherence at their head.  As an example, consider the term:
  \begin{align*}
    &\coh((x : \star)(y : \star)(f : x \to_\star y)(z : \star)(g : y \to_\star z) \\
    &\hspace{.5cm} : \mathsf{comp_2}\sub{x,y,f,z,g} \to_{x \to_\star z} \mathsf{comp_2}\sub{x,y,f,z,g})\\
    &\hspace{.5cm} [x,y,f,z,g]
  \end{align*}
  This term is ``trying'' to be the identity on
  $\mathsf{comp_2}\sub{x,y,f,z,g}$ (indeed, it is provably
  equivalent to it in $\Catt$), but is not a \textit{syntactic}   identity.  Such terms are recognizable as
  coherences for which the type expression has a source and target
  which are equal.  We refer to them as
  \emph{endomorphism coherences}, and our third rule \textsc{Endo} sets
  them equal to the identities they duplicate. In this case, the above term will therefore be definitionally equal to the following:
\[ \i_1\sub{x,z,\mathsf{comp_2}\sub{x,y,f,z,g}} \]
This \textsc{endo} reduction has no apparent analog in the theory of Batanin, Cisinski and Weber~\cite{Batanin2013}, and we believe in that theory the two terms above would not be identified. \label{ex:bcw}

\paragraph{Properties of Definitional Equality} In this section, we record two theorems. The first states that equality interacts nicely with composition of substitutions, which will be crucial when we build the syntactic category in order to define models in Section~\ref{sec:models}.

\begin{restatable}{theorem}{subthm}
  \label{subthm}
  If \(\sigma = \sigma'\) and \(\tau = \tau'\), then \(\sigma \circ \tau = \sigma' \circ \tau'\).
\end{restatable}

\noindent
The second is the following major structural theorem, showing that our definitional equality behaves well with typing.

\begin{restatable}{theorem}{conversionthm}
  \label{thm:conversion-thm}
  Typing is fully compatible with definitional equality, such that the following rules are derivable:
  {\small\begin{mathpar}
    \inferrule{\Gamma \vdash s : A \\ \Gamma \vdash s = t \\ \Gamma \vdash A = B}{\Gamma \vdash t : B} \and
    \inferrule{\Gamma \vdash A \\ \Gamma \vdash A = B}{\Gamma \vdash B}\and

    \inferrule{\Gamma \vdash \sigma : \Delta\\ \Gamma \vdash \sigma = \tau}{\Gamma \vdash \tau : \Delta}
  \end{mathpar}}
\end{restatable}

The rest of this section builds the theorem required to prove these theorems. We start by showing that identities and inferred types are valid:

\begin{restatable}{lemma}{identities}\label{lem:identities}
  Suppose \(\Gamma \vdash \{A,t\} : D^k\). Then \(\Gamma \vdash t : A\). Conversely, if \(\Gamma \vdash t : A\) then \(\Gamma \vdash \{A,t\} : D^k\). Further if \(\Gamma \vdash t : A\) then \(\Gamma \vdash \i_{\dim \, A}\sub{\{A,t\}} : t \to_A t\).
\end{restatable}
\begin{proof}
  The proof is simply a rearrangement of some typing data.
\end{proof}

\begin{restatable}{proposition}{uniquecohtype}
  \label{uniquecohtype}
  If $\Delta \vdash t : A$, then \(A = \ty(t)\) and \(\Delta \vdash t : \ty(t)\).
\end{restatable}
\begin{proof}
  We induct on the structure of the derivation of
  $\Delta \vdash t : A$. There are three cases. If the
  derivation is simply the introduction rule for coherences, then the
  result is immediate, as $\ty(t)$ is the assigned type. Similarly, if the derivation is the introduction rule for variables, then $\ty(t)$ is again the assigned type.

  Otherwise, the proof must be by the use of the conversion rule.
  In this case, we have that \(\Delta \vdash t : B\) and \(B = A\).
  By induction hypothesis we have that \(B = \ty(t)\) and so by transitivity and symmetry of equality we get that \(A = B = \ty(t)\) as required.

  The second part follows from the first by an application of the conversion rule.
\end{proof}

To progress towards a proof of Theorems~\ref{subthm} and~\ref{thm:conversion-thm}, we must show that our equality is well behaved. In particular we need to show that it is well behaved with respect to context extension and substitution.

\begin{restatable}{lemma}{weakening}\label{lem:weakening}
  Suppose \(\Gamma \vdash A\). If \(\Gamma \vdash s = t\) then \(\Gamma , A \vdash s = t\). Similar results hold for equality of types and substitutions. Further, if \(\Gamma \vdash s : B\) then \(\Gamma , A \vdash s : B\). Again similar results hold for typing of types and substitutions.
\end{restatable}
\begin{proof}
  All the results are proved by a simple mutual induction.
\end{proof}

To show that the equality interacts well with substitution, we want that substituting a term or type by equal substitutions produces equal terms (or types) and that the action of substitution preserves equality. The first is a simple induction.

\begin{restatable}{lemma}{eqsubrinv}
  \label{lem:eq-sub-r-inv}
  Let $\sigma, \sigma' : \Sub$ such that $\Gamma \vdash \sigma = \sigma'$.  For $A : \Type$, $t : \Term$,
  and $\tau : \Sub$, we have:
  \begin{equation*}
    \Gamma \vdash A\sub{\sigma} = A\sub{\sigma'}\quad
    \Gamma \vdash t\sub{\sigma} = t\sub{\sigma'}\quad
    \Gamma \vdash \tau \circ \sigma = \tau \circ \sigma'
  \end{equation*}
\end{restatable}
\begin{proof}
  Since substitution on types is given by structural induction (and the
  base case $A = \star$ is trivial), the first equation follows from
  the second. Similarly, the third equation also follows from the second.

  Now, for the second, if $t$ is a variable, then the result is clear by the
  definition of equality on substitutions, which is just equality of
  the comprising terms. On the other hand,
  $t \equiv \coh(\Gamma : A)[\tau]$, then we are reduced to showing
  that $\tau \circ \sigma = \tau \circ \sigma'$ which follows by the induction hypothesis (using the third equation).
\end{proof}

To prove that the action of substitution preserves equality we need the following lemma:

\begin{lemma}
\label{lem:quotientsub}
Composition of substitutions is compatible with taking quotients:
$$(\mu \circ \sigma) \sslash \alpha \equiv (\mu \sslash \alpha) \circ \sigma$$
\end{lemma}
\begin{proof}
  Pruning \(\alpha\) from the substitution simply removes two terms. It does not matter if we apply \(\sigma\) to these terms and then remove them, or simply remove them first.
\end{proof}

\begin{restatable}{lemma}{eqsublinv}
  \label{lem:eq-sub-l-inv}
  Let $A, A' : \Type$, $t, t' : \Term$, and $\sigma, \sigma' : \Sub$. Now suppose we have \(\tau : \Sub\) with \(\Gamma \vdash \tau : \Delta\). Then the following implications hold.
  \begin{alignat*}{3}
    &\Delta \vdash A = A' &\quad&\implies&\quad&\Gamma \vdash A\sub\tau = A'\sub\tau\\
    &\Delta \vdash t = t' &\quad&\implies&\quad&\Gamma \vdash t\sub\tau = t'\sub\tau\\
    &\Delta \vdash \sigma = \sigma' &\quad&\implies&\quad&\Gamma \vdash \sigma\circ\tau = \sigma'\circ\tau\\
    &\Delta \vdash A &\quad&\implies&\quad&\Gamma \vdash A \sub\tau\\
    &\Delta \vdash t : A &\quad&\implies&\quad&\Gamma \vdash t\sub\tau : A \sub\tau\\
    &\Delta \vdash \sigma : \Theta &\quad&\implies&\quad&\Gamma \vdash \sigma \circ \tau : \Theta
  \end{alignat*}
\end{restatable}
\begin{proof}
  All parts of the proof follow by mutual induction on all 6 statements. The only difficult cases are those for the second statement, where the equation \(t = t'\) is one of the generators of the equality.

  Hence, suppose we have
  \[ \Delta \vdash t \equiv \coh(\Gamma : A)[\sigma] = \coh(\Gamma\sslash \alpha :
    A\sub{\pi_{\alpha}})[\sigma\sslash x] \equiv t' \] Then we are
  reduced to showing that
  $ \Gamma \vdash (\sigma \circ \tau)\sslash \alpha = (\sigma\sslash \alpha) \circ
  \tau$ which follows from Lemma~\ref{lem:quotientsub}.

  Next, if we have
  \[ \Delta \vdash \t \equiv \coh(\uD^{n+1} : S^n)[\sigma] = d_{n+1}\sub{\sigma} \equiv t' \]
  we argue as follows:
  \begin{align*}
    \coh(\uD^{n+1} : S^n)[\sigma]\sub{\tau} &\equiv \coh(\uD^{n+1} : S^n)[\sigma \circ \tau] \\
                                            &= d_{n+1}\sub{\sigma \circ \tau} \\
                                            &\equiv d_{n+1}\sub{\sigma}\sub{\tau}
  \end{align*}
  Finally, in the case that
  \[ t \equiv \coh (\Gamma : t \to_A t)[\sigma] = \i_{\dim \, A + 1}\sub{\{A,t\} \circ \sigma} \equiv t' \]
  we obtain:
  \begin{align*}
    \coh (\Gamma : t \to_A t)[\sigma]\sub{\tau} &\equiv \coh (\Gamma : t \to_A t)[\sigma \circ \tau] \\
                                                &= \i_{\dim \, A+1}\sub{\{A,t\} \circ \sigma \circ \tau} \\
                                                &\equiv \i_{\dim \, A+1}\sub{\{A,t\} \circ \sigma}\sub{\tau}
  \end{align*}
  where all syntactic steps follow from Proposition~\ref{prop:catt-is-cat}.
\end{proof}

We now have enough to prove the first theorem:

\subthm*
\begin{proof}
  This result is a corollary of Lemmas~\ref{lem:eq-sub-r-inv}~and~\ref{lem:eq-sub-l-inv}.
\end{proof}

Next we need to show that the constructions used in pruning are well formed.

\begin{restatable}{lemma}{pruningconstruct}
  \label{lem:pruning-construct}
  Suppose \(\Delta \vdash_p\), \(\Gamma \vdash \sigma : \Delta\), and \(\alpha \in \LM(\Delta)\), with \(\alpha\sub\sigma \equiv \i_k\sub\tau\). Then we have \(\Delta \sslash \alpha \vdash \pi_\alpha : \Delta\), \(\Gamma \vdash \sigma \sslash \alpha : \Delta \sslash \alpha\), and \(\Gamma \vdash \sigma = \pi_\alpha \circ \sigma \sslash \alpha\).
\end{restatable}

\begin{proof}
  The proof proceeds by induction on the peak \(p_\alpha\). We begin with \(\pi_\alpha\). To make notation simpler we will prove for any \(d \in \Dyck\, n\) and \(p \in \Peak\, d\) that \(\lceil \mathsf{excise}\,p \rceil \vdash \mathsf{project}\,p : \lceil d \rceil\), and proceed by induction on \(p\):

In the case that the peak is \(\Updownarrow_{\mathsf{pk}} d\, y\, f\), then the validity of the substitution
\[\langle \id_{\lceil d \rceil}, \tm(d), \i_{\dim \, \ty(d)}\sub{\{\ty(d), \tm(d)\}} \rangle\]
follows by the validity of identity substitutions, validity of \(\tm(d)\) and \(\ty(d)\) (for which we omit the proof), validity of identities (Lemma~\ref{lem:identities}), and that weakening preserves validity (Lemma~\ref{lem:weakening}).

If the peak is of the form \(\Uparrow_{\mathsf{pk}}d\,y\,f\,p\), then we need to show that
\[\langle \mathsf{project}\,p, y , f \rangle \]
is valid. This follows from weakening, inductive hypothesis and that \(\ty(\mathsf{excise}\,p) \equiv \ty(d)\sub{\mathsf{project}\,p}\) holds for all \(d\) and \(p\) (proof of this is again omitted).

If the peak is of the form \(\Downarrow_{\mathsf{pk}}d\,p\) then we are done by inductive hypothesis.

We now move onto showing that for all \(d \in \Dyck\, n\), \(p \in \Peak\, d\), and valid \(\sigma : \lceil d \rceil \to \Gamma\) that \(\Gamma \vdash \mathsf{remove}\,p\,\sigma : \lceil \mathsf{excise}\,p \rceil\) and \(\Gamma \vdash \mathsf{project}\,p \circ \mathsf{remove}\,p\,\sigma = \sigma\). We prove these statements by induction on the peak \(p\):

\proofstep When the peak is \(\Updownarrow_{\mathsf{pk}}d\,y\,f\) then the first statement is clear. Suppose we have \(\sigma \equiv \langle \tau, s , t \rangle\). For the second we need to show that:
\[ \langle \id_{\lceil d \rceil} , \tm(d), \i_{\dim \, \ty(d)}\sub{\{\ty(d),\tm(d)\}} \rangle \circ \tau = \langle \tau, s, t \rangle \]
By Proposition~\ref{prop:catt-is-cat}, \(\id_{\lceil d \rceil} \circ \tau \equiv \tau\) so we are left to prove that \(\tm(d)\sub\tau = s\) and \(\i_{\dim \, \ty(d)}\sub{\{\ty(d),\tm(d)\}}\sub\tau = t\). We now use that \(t \equiv f \sub \sigma \equiv \alpha\sub\sigma \equiv \i_{\dim \, B}\sub{\{B,u\}}\) for some term \(u\) and type \(B\).

By Lemma~\ref{lem:identities} and Proposition~\ref{uniquecohtype}, we know that \(\i_{\dim(B)}\sub{\{B,u\}}\) can be given type \(u \to_B u\). Further we know that \(t\) has type \((\tm(d) \to_{\ty(d)} y) \sub\sigma\). As \(t\) and \(\i_{\dim(B)}\sub{\{B,u\}}\) are the same term we must have that \(s \equiv y\sub\sigma = u\), \(\tm(d)\sub\sigma = u\), and \(\ty(d)\sub\sigma = B\). Therefore we immediately get that \(\tm(d)\sub\tau \equiv \tm(d)\sub\sigma = u = s\). Further we have that:
\begin{align*}
  t &\equiv \i_{\dim(B)}\sub{\{B,u\}}\\
    &= \i_{\dim(\ty(d))}\sub{\{\ty(d)\sub\sigma,\tm(d)\sub\sigma\}}\\
    &\equiv \i_{\dim(\ty(d))}\sub{\{\ty(d),\tm(d)\}}\sub\sigma\\
    &\equiv \i_{\dim(\ty(d))}\sub{\{\ty(d),\tm(d)\}}\sub\tau
\end{align*}
as required. The second syntactic equality above comes from Lemma~\ref{lem:typeunfoldingsub}.

\proofstep In the case that the peak is \(\Uparrow_{\mathsf{pk}}d\,y\,f\,p\), the second statement follows easily from weakening and inductive hypothesis. For this proof we introduce the shorthand \(d \sslash p := \mathsf{excise}\,p\). Then, for the first we again assume that \(\sigma \equiv \langle \tau, s, t \rangle\), which means we need to show that (unwrapping some definitions):
\begin{align*}
  \Gamma \vdash \langle \mathsf{remove}\,p\,\tau, s, t \rangle : \lceil d \sslash p \rceil, &(y : \ty(d \sslash p)),\\[-2pt]
                                                                                            &(f : \tm(d \sslash p) \to_{\ty(d \sslash p)} y)
\end{align*}
As we have \(\Gamma \vdash \langle \tau , s , t \rangle : \lceil d \rceil , (y : \ty(d)), (f : \tm(d)  \to_{\ty(d)} y)\) we obtain:
\begin{itemize}[label=-]
\item Firstly that \(\Gamma \vdash \tau : \lceil d \rceil\), which by inductive hypothesis implies that \(\Gamma \vdash \mathsf{remove}\,p\,\tau : \lceil  d \sslash p \rceil\).
\item Secondly that \(\Gamma \vdash s : \ty(d) \sub\tau\). From this we prove using the inductive hypothesis and Lemma~\ref{lem:eq-sub-r-inv}:
  \begin{align*}
    \ty(d) \sub\tau &= \ty(d) \sub {\mathsf{project}\,p \circ \mathsf{remove}\,p\,\tau}\\
                    &\equiv \ty(d) \sub{\mathsf{project}\,p} \sub{\mathsf{remove}\,p\,\tau}\\
                    &\equiv \ty(d\sslash p) \sub{\mathsf{remove}\,p\,\tau}
  \end{align*}
      and so by the conversion rule we have that \(\Gamma \vdash s : \ty(d\sslash p) \sub{\mathsf{remove}\,p\,\tau}\).
    \item Lastly that \(\Gamma \vdash t : (\tm(d) \to_{\ty(d)} y) \sub {\langle \tau , s \rangle}\). By simplification,
      \[(\tm(d) \to_{\ty(d)} y) \sub {\langle \tau , s \rangle} \equiv \tm(d)\sub\tau \to_{\ty(d)\sub\tau} s\]
      As in the last case, \(\ty(d)\sub\tau = \ty(d\sslash p) \sub{\mathsf{remove}\,p\,\tau}\). Similarly \(\tm(d)\sub\tau = \tm(d\sslash p) \sub{\mathsf{remove}\,p\,\tau}\). Lastly we have \(s \equiv y \sub{\langle \mathsf{remove}\,p\,\tau , s \rangle}\) and so with all these parts and the conversion rule we get that:
      \[\Gamma \vdash t : (\tm(d \sslash p) \to_{\ty(d \sslash p)} y) \sub{\langle \mathsf{remove}\,p\,\tau , s \rangle}\]
    \end{itemize}
    By combining all these parts, we have proven what we needed to prove.

    \proofstep If the peak is of the form \(\Downarrow_{\mathsf{pk}}d\,p\) then both statements follow from inductive hypothesis.

  This covers all of the cases and so all three statements in the lemma hold.
\end{proof}

The last component needed for the proof of Theorem~\ref{thm:conversion-thm} is that support conditions are preserved by definitional equality. For this we need the following lemma.

\begin{restatable}{lemma}{support}\label{lem:support}
  If \(\Gamma \vdash s = t\), then \(\supp(s) = \supp(t)\).
\end{restatable}

\begin{proof}
  Define a new set of support-preserving typing judgements \(\vdash_s\), which have the same structural and conversion rules (Figures~\ref{fig:structural-eq}~and~\ref{fig:conversion}) as \(\Catt_{\su}\), but for each generator (Figure~\ref{fig:eq-rules}) we add the condition that the terms on the left and right hand side of the equality have the same support. For example, we will have this modified rule for disc removal:
  \begin{mathpar}
    \inferrule*[Right=disc']{
      n : \N \\ \Gamma \vdash_s \coh(\uD^{n+1} : S^n)[\sigma] : B \\\\ \supp(\coh (D^{n+1} : S^n)[\sigma]) = \supp(d_{n+1}\sub\sigma)
    }{
      \Gamma \vdash_s \coh(\uD^{n+1} : S^n)[\sigma] = d_{n+1}\sub{\sigma}
    }
  \end{mathpar}

  By a simple induction, if \(\Gamma \vdash_s s = t\) then \(\supp(s) = \supp(t)\). In the rest of the proof we prove that the typing judgements \(\vdash\) and \(\vdash_s\) are actually the same.

  It is clear that if a statement is valid under the support-preserving judgement, then it is also valid with the original judgement. Hence we want to show that validity under the original judgements implies validity under the support-preserving judgements. We proceed by a mutual induction on all (original) judgements. Every case is trivial apart from \(\Gamma \vdash s = t \implies \Gamma \vdash_s s = t\) where the equality arises from some generator. We proceed through these in turn.

  \proofstep Suppose we have \(\Gamma \vdash \coh (\Delta : A)[\sigma] = \coh (\Delta \sslash \alpha : A\sub{\pi_\alpha})[\sigma \sslash \alpha]\) from pruning. By assumption we must have \(\Gamma \vdash \coh (\Delta : A) [\sigma] : B\) for some \(B\), and so by inductive hypothesis \(\Gamma \vdash_s \coh (\Delta : A) [\sigma] : B\) holds and \(\Gamma \vdash_s \sigma : \Delta\). Now note that the proof of Lemma~\ref{lem:pruning-construct} only used Lemmas~\ref{lem:weakening}~and~\ref{lem:eq-sub-r-inv}, for which the only property we required of the generators is that they were preserved by weakening. As support is preserved by weakening, we can deduce that this lemma holds for the support-preserving judgements. This means that \(\Gamma \vdash_s \sigma = \pi_\alpha \circ \sigma \sslash \alpha\) and so \(\supp(\sigma) = \supp(\pi_\alpha \circ \sigma \sslash \alpha)\). As \(\pi_\alpha\) is full we obtain:
    \begin{align*}
      \supp(\coh (\Delta : A)[\sigma]) &= \supp(\sigma)\\
                                       &= \supp(\pi_\alpha \circ \sigma \sslash \alpha)\\
                                       &= \supp(\sigma \sslash \alpha)\\
                                       &= \supp(\coh (\Delta \sslash \alpha : A \sub{\pi_\alpha}) [\sigma \sslash \alpha])
    \end{align*}
    and so \(\Gamma \vdash_s \coh(\Delta : A)[\sigma] = \coh(\Delta \sslash \alpha : A \sub{\pi_\alpha})[\sigma \sslash \alpha]\) as required.
  \proofstep Let \(\Gamma \vdash \coh(D^{n+1} : S^n)[\sigma] = d_{n+1}\sub\sigma\) be an instance of disc removal. As in the last case, \(\Gamma \vdash_s \sigma : D^{n+1}\). As the notion of support behaves nicely in terms valid in the support-preserving judgements, we get that the support of \(\sigma\) is the same as the support of its top dimensional term \(d_{n+1}\sub\sigma\), as we would expect. Therefore \(\Gamma \vdash_s \coh(D^{n+1} : S^n)[\sigma] = d_{n+1}\sub\sigma\).
  \proofstep Now if \(\Gamma \vdash \coh(\Delta : s \to_A s)[\sigma] = \i_{\dim(A)}\sub{\{A,t\} \circ \sigma}\) is an instance of endo-coherence removal, we once again have that \(\Gamma \vdash_s \coh(\Delta : s \to_A s)[\sigma] : B\) for some type \(B\). By the typing rules for coherences, we must either have that \(\supp(s) = \partial^-(\Gamma)\) and \(\supp(s) = \partial^+(\Gamma)\), or that \(\supp(s) = \FV(\Gamma)\). As \(\partial^-(\Gamma)\) and \(\partial^+(\Gamma)\) are not equal, we deduce that the second condition must hold, and so \(\{A,s\}\) is full. Therefore:
    \begin{align*}
      \supp(\coh(\Delta : s \to_A s)[\sigma]) &= \supp(\sigma)\\
                                              &= \supp(\{A,s\} \circ \sigma)\\
                                              &= \supp(\i_{\dim(A)}\sub{\{A,s\} \circ \sigma})
    \end{align*}
    and so \(\Gamma \vdash_s \coh(\Delta : s \to_A s)[\sigma] = \i_{\dim(A)}\sub{\{A,s\} \circ \sigma}\).

  Finally, since we have proven that \(\Gamma \vdash s = t\) implies \(\Gamma \vdash_s s = t\), and that \(\Gamma \vdash_s s = t\) implies that \(\supp(s) = \supp(t)\), we have completed the proof.
\end{proof}

We can then finally prove Theorem~\ref{thm:conversion-thm}.

\conversionthm*
\begin{proof}[Proof of Theorem~\ref{thm:conversion-thm}]
  We proceed by mutual induction on the derivation of the equalities (\(\Gamma \vdash s = t\), \(\Gamma \vdash A = B\), and \(\Gamma \vdash \sigma = \tau\)). In fact, it is necessary for our inductive hypothesis to be of the form ``if \(\Gamma \vdash s = t\) and \(\Gamma \vdash A = B\), then \(\Gamma \vdash s : A\) \emph{if and only if} \(\Gamma \vdash t : B\)'' and similar for terms and substitutions. The cases for types and substitutions are easy, and so we run through the cases for terms. We note that it is sufficient to prove the statement ``if \(\Gamma \vdash s = t\) then \(s\) is valid if and only if \(t\) is valid and \(\ty(s) = \ty(t)\)'', from which we can recover the original statement using the conversion rule and Proposition~\ref{uniquecohtype}.

  \proofstep If the equality is from reflexivity on variables, we have nothing to prove.
  \proofstep If the equality is from the transitivity constructor, we use that ``if and only if'' is a transitive relation.
  \proofstep Similarly if the equality arises from the symmetry constructor, we again use that ``if and only if'' is a symmetric relation. Note that this is why we couldn't have used the weaker induction hypothesis of the form ``if \(\Gamma \vdash s = t\) and \(s\) is valid then \(t\) is valid''.
  \proofstep For the structural rule for coherences, we have \(\Gamma \vdash \coh(\Delta : A)[\sigma] = \coh(\Delta : B)[\tau]\) from \(A = B\) and \(\sigma = \tau\). Suppose \(\coh(\Delta : A)[\sigma]\) is valid. Firstly by inductive hypothesis we have \(\Delta \vdash B\) and \(\Gamma \vdash \tau : \Delta\). By Lemma~\ref{lem:support}, we have that the relevant support condition for \(\coh(\Delta : B)[\tau]\) is satisfied and so this term is valid. The other direction follows symmetrically. Finally we also have:
    \[\ty(\coh(\Delta : A)[\sigma]) \equiv A\sub\sigma = B \sub \tau\equiv \ty(\coh(\Delta : B)[\tau])\]
    where the equality uses Lemmas~\ref{lem:eq-sub-r-inv}~and~\ref{lem:eq-sub-l-inv}.
  \proofstep Suppose we have \(\Gamma \vdash \coh(\Delta : A)[\sigma] = \coh(\Delta\sslash \alpha : A\sub{\pi_\alpha})[\sigma\sslash \alpha]\) from pruning. The backwards direction is clear so we prove the forwards direction. Suppose \(\Gamma \vdash \coh(\Delta : A)[\sigma]\) is valid, which implies \(\Gamma \vdash \sigma : \Delta\) and \(\Delta \vdash A\). Therefore by Lemma~\ref{lem:pruning-construct}, \(\Delta \sslash \alpha \vdash \pi_\alpha : \Delta\) and \(\Gamma \vdash \sigma \sslash \alpha : \Delta \sslash \alpha\). By Lemma~\ref{lem:eq-sub-l-inv}, \(\Delta \vdash A \sub{\pi_\alpha}\).

    All that remains to show that the right hand side of the equation is valid is that the appropriate support condition holds. Assume that \(A \equiv s \to_{A'} t\). First, assume that \(\supp(s) = \supp(t) = \FV(\Delta)\). Then,
    \[\supp(s\sub{\pi_\alpha}) = \supp(t\sub{\pi_\alpha}) = \supp(\pi_\alpha) = \FV(\Delta \sslash \alpha)\]

    Otherwise \(\supp(s) = \partial^-(\Delta)\) and \(\supp(t) = \partial^+(\Delta)\). We first suppose that the dimension of \(\Delta \sslash \alpha\) is one less than the dimension of \(\Delta\). In this case, \(\alpha\) was the only variable of maximum dimension and so \(\supp(s)\) is the whole of \(\Delta\) except \(\alpha\) and \(\tgt(\alpha)\) and \(\supp(t)\) is the whole of \(\Delta\) except \(\alpha\) and \(\src(\alpha)\). By the definition of \(\pi_\alpha\) it is immediate that \(\supp(s\sub{\pi_\alpha}) = \supp(t\sub{\pi_\alpha}) = \FV(\Delta\sslash\alpha)\), which is sufficient for validity.

    The last case is as above but assuming that the dimension of \(\Delta\) and \(\Delta\sslash\alpha\) are the same. Then \(\partial^-(\Delta\sslash\alpha) = \partial^-(\Delta)\) as they both have all lower dimensional variables, no maximal dimensional variables and all target-free variables, which are the same in both. Since \(\pi_\alpha\) is the identity of \(\partial^-(\Delta)\) we have that \(\supp(s\sub{\pi_\alpha}) = \partial^-(\Delta\sslash\alpha)\). Similarly we have that \(\partial^+(\Delta\sslash\alpha) = \partial^+(\Delta)\setminus \{\tgt(\alpha)\} \cup \{\src(\alpha)\}\) by inspecting which variables are source free. Then as \(\pi_\alpha\) is the identity on \(\partial^+(\Delta) \setminus\{\tgt(\alpha)\}\) and sends \(\tgt(\alpha)\) to \(\src(\alpha)\) we have that \(\supp(t) = \partial^+(\Delta\sslash\alpha)\) as required.

    Finally we have:
    \begin{align*}
      \ty(\coh(\Delta : A)[\sigma]) &\equiv A \sub \sigma\\
                                    &= A \sub{\pi_\alpha \circ \sigma \sslash \alpha}\\
                                    &\equiv A \sub{\pi_\alpha}\sub{\sigma\sslash \alpha}\\
                                    &\equiv \ty(\coh(\Delta \sslash \alpha : A\sub{\pi_\alpha})[\sigma\sslash \alpha])
    \end{align*}
    from Lemmas~\ref{lem:eq-sub-r-inv}~and~\ref{lem:pruning-construct}.
  \proofstep For disc removal, we again only need to show the forward direction. If \(\coh(\uD^{k+1} : S^{k})[\sigma]\) is valid, then \(\sigma\) is too. Since \(\sigma\) is a substitution from a disc, \(\sigma \equiv \{S^k\sub\sigma,d_{k+1}\sub\sigma\}\), and by Lemma~\ref{lem:identities}, we have \(\Gamma \vdash d_{k+1}\sub\sigma : S^k \sub\sigma\). Further, by Proposition~\ref{uniquecohtype}, we must have \(\ty(d_{k+1}\sub\sigma) = S^k\sub\sigma \equiv \ty(\coh(\uD^{k+1} : S^k)[\sigma])\).
  \proofstep For endo-coherence removal, we once again only consider the forwards direction. If \(\coh(\Delta : s \to_A s)[\sigma]\) is valid in \(\Gamma\), then \(\Gamma \vdash \sigma : \Delta\) and \(\Delta \vdash s : A\). Therefore by Lemma~\ref{lem:identities} \(\Delta \vdash \{A,s\} : \uD^{\dim(A) + 1}\) and so \(\Gamma \vdash \{A,s\} \circ \sigma : \uD^{\dim(A) + 1} \). This implies that \(\i_{\dim(A)}\sub{\{A,s\} \circ \sigma}\) is valid and by Lemma~\ref{lem:typeunfoldingsub} we have \(\i_{\dim(A)}\sub{\{A,s\} \circ \sigma} \equiv \i_{\dim(A)}\sub{\{A\sub\sigma,s\sub\sigma\}}\) and so
    \begin{align*}
      \ty(\coh(\Delta : s \to_A s)[\sigma]) &\equiv s\sub\sigma \to_{A\sub\sigma} s\sub\sigma\\
                                            &\equiv \ty(\i_{\dim(A)}\sub{\{A\sub\sigma,s\sub\sigma\}})\\
                                            &\equiv \ty(\i_{\dim(A)}\sub{\{A,s\} \circ \sigma})
    \end{align*}
    which is what we needed to prove.

  The induction hypothesis is a strengthening of the statement of the theorem, and so we are done.
\end{proof}

\subsection{Models of \texorpdfstring{$\Catt$}{Catt} and \texorpdfstring{$\Catt_{\su}$}{Cattsu}}
\label{sec:models}

The type theories $\Catt$ and $\Catt_{\su}$ generate syntactic
categories via a standard construction: objects are contexts, and
morphisms are substitutions, up to definitional equality. Composition is given by composition
of substitutions, which is well-defined (for $\Catt_{\su}$) by Theorem~\ref{subthm}, and associative and unital (for both theories) by Proposition~\ref{prop:catt-is-cat}.  We abuse notation slightly and write $\Catt$ and
$\Catt_{\su}$ for the corresponding categories.  We also
write $\Catt^{\pd}$ and $\Cattsupd$ for the full subcategories
consisting of just the pasting contexts.

As we have seen, this category contains a collection of objects
$\uD^k$ corresponding to the $k$-dimensional disc contexts.  Together, these contexts and their source and target substitutions
constitute a \emph{globular object} in $\Catt$ and $\Catt_{\su}$, respectively, i.e. a diagram \[
\begin{tikzpicture}[xscale=2, yscale=.8]
\node (0) at (0,0) {$\cdots\vphantom{D^k}$};
\node (1) at (1,0) {$\uD^{k+1}$};
\node (2) at (2,0) {$\uD^{k}$};
\node (3) at (3,0) {$\cdots\vphantom{D^k}$};
\node (4) at (4,0) {$\uD^0$};
\draw [bend left, ->] (0.north east) to node[above] {$\sigma$} (1.north west);
\draw [bend left, ->] (1.north east) to node[above] {$\sigma$} (2.north west);
\draw [bend left, ->] (2.north east) to node[above] {$\sigma$} (3.north west);
\draw [bend left, ->] (3.north east) to node[above] {$\sigma$} (4.north west);
\draw [bend right, ->] (0.south east) to node[below] {$\tau$} (1.south west);
\draw [bend right, ->] (1.south east) to node[below] {$\tau$} (2.south west);
\draw [bend right, ->] (2.south east) to node[below] {$\tau$} (3.south west);
\draw [bend right, ->] (3.south east) to node[below] {$\tau$} (4.south west);
\end{tikzpicture}
\] such that $\sigma \circ \sigma = \sigma \circ \tau$ and $\tau \circ \sigma = \tau \circ \tau$.
\begin{definition}
  A category $\mathcal{C}$ containing a globular object is said to
  admit \emph{globular limits} if every diagram of the following form admits a limit, where $j_k \leq i_{k-1}, i_k $ for all $1\leq k \leq n$:
  \[\begin{tikzpicture}[xscale=2, scale=0.9]
  \node (0) at (0,1) {$D^{i_0}$};
  \node (1) at (2,1) {$D^{i_1}$};
  \node (11) at (1,0) {$D^{j_1}$};
  \node (12) at (3,0) {$D^{j_n}$};
  \node (2) at (4,1) {$D^{i_n}$};
  \draw[->] (0) to node[below left, scale=0.7] {$\underbrace{\tau \circ \cdots \circ \tau}_{i_0 - j_1}$}(11);
  \draw[->] (1) to node[below right, scale=0.7] {$\underbrace{\sigma \circ \cdots \circ \sigma}_{i_1 - j_1}$}(11);
  \draw[->] (2) to node[below right, scale=0.7] {$\underbrace{\sigma \circ \cdots \circ \sigma}_{i_n - j_n}$}(12);
  \path (1) to node {$\cdots$} (12);
  \end{tikzpicture}
  \]
%
  Dually, a category $\mathcal{C}$ containing a
  co-globular object is said to admit \emph{globular sums} if the
  category $\mathcal{C}^{\op}$ admits globular limits.
\end{definition}

\begin{theorem}
\mbox{$\Catt$ and $\Catt_{\su}$ admit globular limits.}
\end{theorem}

\noindent
The proof that \Catt admits globular limits is to some extent folklore, and has recently been written out~\cite{Benjamin}. The proof for $\Catt_\su$ is the same, as it depends only on the variable structure of pasting contexts, which is the same in both theories.

With this in hand, we can give our notion of model.
\begin{definition}
  An \emph{$\infty$-category} is a presheaf on the category
  $\Catt^{\pd}$ which sends globular limits to
  globular sums. A \emph{strictly unital $\infty$-category}
  is a presheaf on $\smash{\Catt^{\pd}_{\su}}$ which sends globular limits to
  globular sums.
\end{definition}

\noindent
We write $\cCat_{\infty}$ for the resulting category of
$\infty$-categories and $\cCat^{\mathrm{su}}_{\infty}$ for the
strictly unital $\infty$\-categories, each being full subcategories of
the category of presheaves on $\Catt^\pd$ and $\Cattsupd$
respectively.

\smallskip

Note that since every valid $\Catt$ term will also be a valid
$\Catt_\su$ term, there is an evident functor
\mbox{$K: \Catt^\pd \to \Cattsupd$} which is the identity on objects (pasting
diagrams are the same in both theories), and which sends substitutions to
their equivalence classes under definitional equality.  Since $K$
preserves globular limits, precomposition with
\mbox{$K$} determines a  functor \mbox{$K^* : \cCat^{\mathrm{su}}_{\infty} \to \cCat_{\infty}$}, sending a strictly
unital $\infty$-category to its underlying $\infty$-category. We will
see in Section~\ref{sec:rehydration} that this functor is
fully-faithful; that is, that every $\infty$-category can arise from
at most one strictly unital $\infty$\-category in this way. In other
words, being strictly unital is a \emph{property} of an
$\infty$\-category.

\section{Reduction}
\paragraph{Overview}In this section we introduce a reduction relation~$\squig$ on types, terms and substitutions, and show that its reflexive, transitive and symmetric closure agrees with definitional equality. We then define a subrelation called \textit{standard reduction}, written $\bsquig$, and show it is a partial function which terminates after finitely many steps, giving a notion of normal form. Finally, we show that two terms have the same normal form just when they are definitionally equal, meaning that standard reduction gives an algorithm for deciding definitional equality.

\paragraph{Convention on Contexts} We say ``$U$ is valid in $\Gamma$'' to mean $\Gamma \vdash U$. We say ``$t$ is valid in $\Gamma$'' to mean there is some type $U$ such that $\Gamma \vdash t:U$. We say ``$t$ is valid'' to mean there exists some context $\Gamma$ and type $U$ such that $\Gamma \vdash t:U$.

When a fixed context $\Gamma$ is under consideration, we will often write $s = t$ to mean $\Gamma \vdash s = t$, and will do similar for types and substitutions.

\paragraph{Reflexive, Transitive, Symmetric Closures}Given a relation $\squig$ between terms, types, or substitutions, we write $\squig_\r$ for its reflexive closure, $\squig_\t$ for its transitive closure, and $\squig_\s$ for its symmetric closure,  in the set of valid syntactic entities. When we use multiple such subscripts, we mean this simultaneously; for instance, we write  $\squig_\rts$ for its simultaneous reflexive, transitive and symmetric closure of valid entities.

\subsection{General reduction}

We define a reduction relation on types, terms and substitutions, and show that the equivalence relation generated by these relations agrees with definitional equality.

We first define a simple syntactic property on terms, that of being an identity.

\begin{definition}\label{def:identityterm}
A term is an \emph{identity} if it is of the form $\i_n \sub{\sigma}$ for some $n \in \N$; that is, when its head is an identity coherence.
\end{definition}

\noindent
We emphasise that as a syntactic property, this is \textit{not} compatible with definitional equality. For example, if \mbox{$t \equiv \i_n \sub{\sigma}$}, then $t$ is an identity; but if we merely have $t = \i_n\sub{\sigma}$, then~$t$ is not necessarily an identity.

We now give the reduction relation. The following definitions are given by simultaneous induction. Determining whether a given type, term or substitution is a redex is a purely \textit{syntactic} condition, which can be mechanically verified, and does not refer to definitional equality.
\begin{definition}[Reduction of types]
The basic type $\star$ does not reduce. An arrow type \mbox{$U \equiv (u \to_T v)$} reduces as follows:
\begin{enumerate}
\item[(T1)] if $u \squig u'$, then:
$(u \to _T v) \xsquig {T1} (u' \to_T v)$
\item[(T2)] if $v \squig v'$, then:
\((u \to _T v) \xsquig {T2} (v \to _T v')\)
\item[(T3)] if $T \squig T'$, then:
\((u \to _T v) \xsquig {T3} (u \to _{T'} v)\)
\end{enumerate}
\end{definition}

\begin{definition}[Reduction of substitutions] A substitution $\sigma \equiv [ s_1, \ldots, s_n ]$ reduces as follows, given a reduction $s_i^{} \squig s_i'$ of some argument:
$$[ s_1, \ldots, s_i, \ldots, s_n ] \xsquig S [ s_1, \ldots, s_i', \ldots, s_n ]$$
\end{definition}

\begin{definition}[Reduction of terms]
\label{def:redterms}
Variable terms do not reduce. A coherence \mbox{$t \equiv \coh (\Gamma:T)[\sigma]$} reduces as follows:
\begin{enumerate}[label=(\Alph*)]
\item if $\sigma \squig \sigma'$, then:
\begin{align*}
\coh (\Gamma : T)[\sigma] &\xsquig A \coh(\Gamma : T)[\sigma']
\intertext{\item if $t$ is not an identity, and $x \in \LM(\Gamma)$ for which $x[\sigma]$ is an identity, then we define:}
\coh(\Gamma : T)[\sigma]
&\xsquig B \coh(\Gamma \sslash x : T\sub{\pi_x})[\sigma \sslash x]
\intertext{\item[(C)] if $T \squig T'$, then:}
\coh (\Gamma : T)[\sigma]
&\xsquig C
\coh(\Gamma: T')[\sigma]
\intertext{\item[(D)] the disc removal relation:}
\coh(\uD^{n+1} : S^n)[\ldots, t]
&\xsquig D
t
\intertext{\item[(E)] if $t$ is not an identity, the endomorphism coherence removal relation:}
\coh(\Gamma : u \to_T u)[\sigma] &\xsquig E \i_{\dim \, T+1} \sub{\{T,u\} \circ \sigma }
\end{align*}
\end{enumerate}
\end{definition}

\noindent
If we can reduce $u \squig u'$ via some reduction stage  (X) above, we say that $u$ is a \emph{general X\-redex}, or just an \emph{X-redex}, and we write $u \xsquig X u'$. A given term can be a general X-redex for multiple stages (X). For example, if $u \squig u'$, then the term $\coh (\Gamma : u \to_T u)[\sigma]$ is a C-redex in at least 2 ways, and also an E-redex, as follows:
\begin{align*}
\coh(\Gamma : u \to_T u)[\sigma] &\xsquig C \coh(\Gamma : u' \to_T u)[\sigma]
\\
\coh(\Gamma: u \to_T u)[\sigma] &\xsquig C \coh(\Gamma:u \to_T u')[\sigma]
\\
\coh(\Gamma: u \to_T u)[\sigma] &\xsquig E \i_{\dim \, T + 1} \sub{ \{ T,u \} \circ \sigma }
\end{align*}
A term could also have no reductions at all. So reduction is partially defined, and multivalued in general.

We now show that the equivalence relation generated by this reduction relation agrees with definitional equality constructed in Section~\ref{sec:typetheory}.

\begin{restatable}{proposition}{equalityreduction}
\label{equalityreduction}
For any context $\Theta$, and for any terms, types, and substitutions valid in $\Theta$, the equality relation $=$ agrees with $\squig_\rts$, where $\squig_\rts$ is the reflexive transitive symmetric closure of $\squig$ in the set of terms, types, or substitutions valid in~\(\Theta\).
\end{restatable}
\begin{proof}
  To prove \(\squig_\rts\) implies \(=\), it is sufficient to prove that \(s \squig t\) implies \(s = t\) for terms \(s\) and \(t\) valid in \(\Theta\), as \(=\) is an equivalence relation and the statement for types and substitutions reduces to the statement for terms. This statement follows from a simple induction, noting that subterms of valid terms are themselves valid.

We say that a term $p$ is \emph{conservative} if for all contexts \(\Theta\) and terms $q$ such that $p$ and $q$ are valid in \(\Theta\), we have that $p = q$ implies $p \squig_\rts q$. To prove the lemma, we must therefore show that all terms are conservative. Our proof operates by induction on subterms of $p$, and by case analysis on the equality $p=q$. Almost all such cases are immediate; here we explicitly handle the only non-trivial case.

Suppose $p=q$ is the following equality, obtained by pruning an identity term:
\begin{align*}
&\i_n \sub{\ldots, v_2, v^{}_1, v_1', \i_{n-1} \sub{\ldots, u_3, u_2, u_1}} \\
&= (\coh \uD^{n-1}:\i_{n-1} \to_{(d_{n-1} \to _{S^{n-2}} d_{n-1})} \i_{n-1})[\ldots, v_2, v_1]
\end{align*}
By induction on subterms of $p$ we conclude that the terms $u_i$, $v_1'$, $v_i$ are conservative. By validity we must have $u_1 = v_1 = v_1'$ and $u_i = v_i$, and hence we conclude $u_1 \squig_\rts v_1 \squig_\rts v_1'$ and $u_i \squig_\rts v_i$. We do not have $p \xsquig B q$, since $p$ is an identity term, which are explicitly proscribed as B-redexes. However, $q$ admits an E\-reduction, and then a further series of A-reductions obtained by conservativity of the subterms of $p$, as follows:
\begin{align*}
q &\xsquig E \i_n \sub{\ldots, v_3, v_2, v_1, v_1, \i_{n-1} \sub{\ldots, v_3, v_2, v_1 } }
\\
&\xsquig A _\rts \i_n \sub{\ldots, v_3, v_2, v_1^{}, v_1', \i_{n-1} \sub{\ldots, u_3, u_2, u_1 } }
\\&\equiv p
\end{align*}
Hence $q \squig_\rts p$ as required.
\end{proof}

\noindent
From this we get the following lemma as a corollary.

\begin{restatable}{corollary}{reductionvalidity}
\label{cor:reductionvalidity}
Reduction of types, terms and substitutions preserves judgement validity:
\begin{itemize}
\item if $\Gamma \vdash A$ is valid and $A \squig A'$, then $\Gamma \vdash A'$ is valid;
\item if $\Gamma \vdash t :A$ is valid and $t \squig t'$, then $\Gamma \vdash t' : A$ is valid;
\item if $\Gamma \vdash t : A$ is valid and $A \squig A'$, then $\Gamma \vdash t : A'$ is valid;
\item if $\Gamma \vdash \sigma : \Delta$ is valid and $\sigma \squig \sigma'$, then $\Gamma \vdash \sigma' : \Delta$ is valid.
\end{itemize}
\end{restatable}
\begin{proof}
Immediate from Proposition~\ref{equalityreduction} and Theorem~\ref{thm:conversion-thm}.
\end{proof}

\subsection{Standard reduction}

We now define \textit{standard reduction}, denoted by $\bsquig$, a sub-relation of general reduction $\squig$. Standard reduction is a reduction strategy, in the following sense.

\begin{definition}
A \emph{reduction strategy} is a relation $\to$ on terms with the property that, if $a \to b$ and $a \to b'$, then $b \equiv b'$.
\end{definition}

Standard reduction works in a similar way to general reduction, but the reductions now have a preference order, so that higher-priority redexes, listed earlier in the following list, block lower-priority redexes listed later. Standard reduction is hence a reduction strategy by construction.

\begin{definition}[Standard reduction of types]
The standard reduction of a type $U \equiv (u \to_T v)$ is given by the first matching reduction in the following list, if any:
\begin{enumerate}
\item[(T0)] if $T \bsquig \widetilde T$, then:
\((u \to _{T\vphantom{\widetilde T}} v) \xbsquig {T0} (u \to _{\widetilde T} v)\)
\item[(T1)] if $u \bsquig \tilde u$, then:
\((u \to _T v) \xbsquig {T1} (\tilde u \to_T v)\)
\item[(T2)] if $v \bsquig \tilde v$, then:
\((u \to _T v) \xbsquig {T2} (v \to _T \tilde v)\)
\end{enumerate}
\end{definition}

\begin{definition}[Standard reduction of substitutions] Given a  substitution $\sigma \equiv \langle s_1, \ldots, s_n \rangle$, then if $s_i \bsquig \tilde s_i$ is the leftmost argument with a standard reduction, we have the following:
$$[ s_1, \ldots, s_i, \ldots, s_n] \xbsquig {S} [ s_1, \ldots, \tilde s_i, \ldots, s_n ]$$
\end{definition}

\begin{definition}[Standard reduction of terms]
\label{standardreductionterms}
A coherence term $t \equiv \coh (\Gamma:U)[\sigma]$ has a \textit{standard reduction} given by the first listed reduction which is defined, if any:
\begin{enumerate}[label=(\Alph*)]
\item if $\sigma \bsquig \tilde \sigma$, then:
\begin{align*}
\coh (\Gamma : U)[\sigma] &\xbsquig A \coh(\Gamma : U)[\tilde \sigma]
\intertext{\item if $t$ is not an identity, and $x \in \var(\Gamma)$ is the leftmost locally-maximal variable for which $x[\sigma]$ is an identity, then we define:}
\coh(\Gamma : U)[\sigma]
&\xbsquig B \coh(\Gamma \sslash x : U\sub{\pi_x})[\sigma \sslash x]
\intertext{\item[(C)] if $T \bsquig \tilde T$, then:}
\coh (\Gamma : T)[\sigma]
&\xbsquig C
\coh(\Gamma: \tilde T)[\sigma]
\intertext{\item[(D)] the disc removal relation:}
\coh(\uD^{n+1} : S^n)[\ldots, t]
&\xbsquig D
t
\intertext{\item[(E)] if $t$ is not an identity, the endo-coherence removal relation:}
\coh(\Gamma : u \to_A u)[\sigma] &\xbsquig E \i_{\dim \, A + 1} \sub{\{A,u\} \circ \sigma}
\end{align*}
\end{enumerate}
\end{definition}

\noindent
If we can reduce $s \bsquig t$ via some reduction label (X) above, we say that $u$ is a \emph{standard X\-redex}. It is an immediate consequence of the definition of standard reduction that it is a reduction strategy; that is, if a term, type or substitution has a standard reduction, it has \textit{exactly one} standard reduction. This is quite  unlike general reduction as defined as above. For example, suppose $u \bsquig \tilde u$, and consider the term $t \equiv \coh (\Gamma : u \to_A u)[\sigma]$. It is possible that $t$ is a standard A\-redex; failing that, it could be a standard B-redex; failing that, it will certainly be a standard C-redex. Although $t$ is an E-redex (that is, there exists $t'$ with $t \xsquig E t'$), it is {not} a \textit{standard} E-redex, since standard C-reductions are higher-priority than standard E-reductions.

Since standard reduction is unique when it exists, it is useful to introduce the following notation.
\begin{definition}
If $s$ has a standard reduction, we write it as $\tilde s$, and hence $s \bsquig \tilde s$. We call $\tilde s$ the \emph{standard reduct} of $s$.
\end{definition}

\noindent
These simple lemmas now follow.

\begin{restatable}
  {lemma}{standardtogeneral}\label{lem:standardtogeneral}If $s \bsquig \tilde s$, then $s \squig \tilde s$.\end{restatable}%
\begin{proof}
By definition, standard reduction is a subrelation of general reduction.
\end{proof}

\vspace{-5pt}

\begin{restatable}
{lemma}{generaltostandard}
\label{lem:generaltostandard}
If $s \squig t$, then there is a unique $\tilde s$ with $s \bsquig \tilde s$.
\end{restatable}
\begin{proof}
Since standard reduction is a reduction strategy, uniqueness is clear. What we must establish is existence. The intuition is straightforward: in essence, we define standard reduction by giving a priority order to the redexes for general reduction, and allowing only the highest-priority redex. The result is then immediate, because if $s$ has at least one reduction, then there must be a highest-priority such reduction.

We prove the result formally as follows, by simultaneous induction on the structure of terms, types, and substitutions. For the base cases, given by the type $\star$ or a variable term, there is no reduction, so the claim is vacuously true.

For a compound type $U\equiv (a\to_T b)$ is a type, then the statement follows immediately by induction on $a$, $b$ or $T$.

For a substitution $\sigma = \langle s_1,\ldots, s_n \rangle$, suppose we have some reduction $\sigma \squig \tau$ arising from some choice of index $i$ and some reduction $s_i \squig t$. Since $\sigma$ has a reducible argument, it must have a leftmost reducible argument, which we can write as $s_j$, with $j \leq i$. By induction on subterms $s_j \bsquig \tilde s_j$, and $[ \ldots, s_{j-1}, s_j, s_{j+1}, \ldots ] \bsquig [ \ldots, s_{j-1}, \tilde s_j, s_{j+1}, \ldots ]$ is the required standard reduction.

For a coherence term $s \equiv \coh (\Gamma : T)[\sigma]$ with a reduction $s \squig t$, we argue by case analysis as follows.
\begin{itemize}
\item
If $u$ is an A-redex, there must exist some  $\sigma'$ such that $\sigma \leadsto \sigma'$. By induction $\sigma \bsquig \tilde \sigma$, and hence $\coh (\Gamma:T)[\sigma] \xbsquig A \coh(\Gamma:T)[\tilde \sigma]$.

\item If $u$ is not an A-redex, but $u$ is a B-redex, then there must be some leftmost locally-maximal argument of $\Gamma$ with respect to which it is a standard B\-redex.

\item If $u$ is not an A- or B-redex, but it is a C-redex, then there must exist some $T'$ such that $T \squig T'$. Hence by induction $T \bsquig \tilde T$, and so $u \xbsquig C \coh (\Gamma : \tilde T)[\sigma]$.

\item If $u$ is not an A-, B- or C-redex, but it is a D-redex, then the D-reduction will be standard.

\item If $u$ is not an A-, B-, C- or D-redex, but it is an E-redex, then the E\-reduction will be standard.

\item If $u$ is not an A-, B-, C-, D- or E-redex, then $u$ cannot be reduced, contradicting the hypothesis of the theorem.
\end{itemize}
This completes the argument.
\end{proof}

\subsection{Termination of standard reduction}

Standard reduction gives us a notion of normal form.
\begin{definition}
We define a term, type or substitution to be in \emph{normal form} when it has no reduction, either by standard or general reduction (by Lemmas~\ref{lem:standardtogeneral} and~\ref{lem:generaltostandard}, these conditions are equivalent.)
\end{definition}

\noindent
In this section we show that standard reduction terminates after finite time on every term, type and substitution. This means that for every term $t$ we can obtain a term $N(t)$ in normal form, by repeatedly applying standard reduction until a normal form is reached.

To work towards our termination result, we consider what happens when we start with a given term and repeatedly perform standard reduction. What we find is a specific pattern of standard reductions, illustrated in Figure~\ref{fig:standardreduction}. We first perform some family of standard A-reductions. If these terminate, they will be followed by some family of standard B-reductions. If these terminate, they will be followed by some family of standard C\-reductions. If these terminate, and they are not yet sufficient to yield a normal form, they will be followed either by a single standard D-reduction, giving a term in normal form;  or by a single standard E-reduction and a family of standard A\-reductions, which if they terminate will yield a normal form. We prove this claim here.

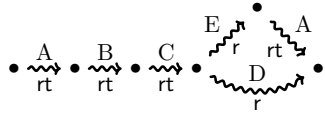
\begin{figure}[t!]
\vspace{-.5cm}

\[
\scalebox{0.8}{\begin{tikzpicture}
\node (1) at (0,0) {$\bullet$};
\node (2) at (1,0) {$\bullet$};
\node (3) at (2,0) {$\bullet$};
\node (4) at (3,0) {$\bullet$};
\node (5) at (4,1) {$\bullet$};
\node (6) at (5,0) {$\bullet$};
\draw [standard] (1) to node [above] {A} node [below] {\rt} (2);
\draw [standard] (2) to node [above] {B} node [below] {\rt} (3);
\draw [standard] (3) to node [above] {C} node [below] {\rt} (4);
\draw [standard] (4) to [bend right] node [above] {D} node [below] {\r} (6);
\draw [standard] (4) to node [above left=-1pt] {E} node [below right=-2pt] {\r} (5);
\draw [standard] (5) to node [above right=-1pt] {A} node [below left=-2pt] {\rt} (6);
\end{tikzpicture}}
\]

\vspace{-8pt}
\caption{Standard reduction pathways to normal form}
\label{fig:standardreduction}
\end{figure}

\begin{restatable}{proposition}{progressive}
\label{prop:progressive}
The reflexive transitive closure of the standard reduction relation is obtained as the following composite:
$$\bsquig_\rt \quad= \quad
\xbsquig {A}_\rt
\,\xbsquig {B}_\rt
\,\xbsquig {C}_\rt
\Big(
\xbsquig {D}_\r
\cup
\,\xbsquig {E}_\r
\,\xbsquig {A}_\rt
\Big)
$$
\end{restatable}
\begin{proof}
We establish the claim by showing that the following composite reduction pairs and triples are impossible (that is, they are empty as relations).
\begin{itemize}
\item $p \xbsquig B q \xbsquig A r$. For $p$ to be a standard B--redex, it is required that it is not an A-redex, and so the arguments of $p$ must be in normal form. But the arguments of $q$ are a subset of the arguments of $p$, contradicting the claim that $q$ is a standard A-redex.
\item $p \xbsquig C q \xbsquig A r$. For $p$ to be a standard C-redex, it is required that it is not an A-redex, and so the arguments of $p$ must be in normal form. But C\-reductions do not change the arguments of a term, so the arguments of $q$ are the same as the arguments $p$, contradicting the claim that $q$ is a standard A-redex.
\item $p \xbsquig C q \xbsquig B r$. For $p$ to be a standard C-redex, it is required that it is not a B-redex, a condition that depends on the context and arguments of the term. But C-reductions do not change the context or arguments, so it is impossible that $q$ is a B-redex.
\item $p \xbsquig D q \bsquig r$. For a term to be a standard D-redex, it is required that it is not an A-redex, and so the arguments of $p$ must not be reducible. By the action of D-reduction, the term $q$ is one of the arguments of $p$, contradicting the claim that $q$ is reducible.
\item $p \xbsquig E q \xbsquig B r$. The standard E-reduct $q$ is an identity term, but identity terms are never B\-redexes, by definition.
\item $p \xbsquig E q \xbsquig C r$. The standard E-reduct $q$ is an identity term, but identity terms are never C\-redexes, since the cell part is in normal form.
\item $p \xbsquig E q \xbsquig D r$. The standard E-reduct $q$ is an identity term, but identity terms are never D\-redexes, as the head has the wrong syntactic form.
\item $p \xbsquig E q \xbsquig E r$. The standard E-reduct $q$ is an identity term, but identity terms are never E\-redexes, by definition.
\item $p \xbsquig E q \xbsquig A_\rt r \xbsquig B s$. The standard E-reduct $q$ is an identity term. A\-reductions do not change the head, and so $r$ will also be an identity term. But identity terms are never B-redexes, giving a contradiction.
\item $p \xbsquig E q \xbsquig A _\rt r \xbsquig C s$. The standard E-reduct $q$ is an identity term. A\-reductions do not change the head, and so $r$ will also be an identity term. But identity terms are never C-redexes, giving a contradiction.
\item $p \xbsquig E q \xbsquig A _\rt r \xbsquig D s$. The standard E-reduct $q$ is an identity term. A\-reductions do not change the head, and so $r$ will also be an identity term. But identity terms are never D-redexes, giving a contradiction.
\item $p \xbsquig E q \xbsquig A _\rt r \xbsquig E s$. The standard E-reduct $q$ is an identity term. A\-reductions do not change the head, and so $r$ will also be an identity term. But identity terms are never E-redexes, by definition.
\end{itemize}
The result is then established as follows, by imagining a standard reduction sequence for some given coherence term. Here we refer to composite relations by concatenation; so for example, $t \xbsquig B \xbsquig A t'$ just when there exists some $t''$ with $t \xbsquig B t'' \xbsquig A t'$.
\begin{itemize}
\item Standard A-reductions have the highest priority, so these will be performed first.
\item If the above step terminates, standard B-reductions have the second-highest priority, so we now perform these. Since $\xbsquig B \xbsquig A$ is empty, this will not trigger any additional standard A\-reductions.
\item If the above step terminates, standard C-reductions have the third-highest priority, so we now perform these. Since $\xbsquig C \xbsquig A$ and $\xbsquig C \xbsquig B$ are both empty, these standard C-reductions will not trigger further standard A- or B-reductions.
\item If the above step terminates, standard D-reductions have the fourth-highest priority. If we can perform a D-reduction, the result will be in normal form, since $\xbsquig D \bsquig$ is empty.
\item If we cannot perform a D-reduction, we consider applying a standard E\-reduction, as the standard reduction with fifth-highest priority. If the standard E-reduction cannot be applied, then the term is in normal form, as the standard E-reduction is the last reduction in the list.
\item If the standard E-reduction was successfully applied, it cannot be followed by a standard B-, C-, D-, or E- reduction, since $\xbsquig E \xbsquig B$, $\xbsquig E \xbsquig C$, $\xbsquig E \xbsquig D$ and $\xbsquig E \xbsquig E$ are all shown above to be empty. The only remaining possibility is for the standard E-reduction to be followed by some sequence of standard A-reductions. These A\-reductions cannot themselves be followed by a standard B-, C-, D- or E-reduction, since we show above that $\xbsquig E \xbsquig A _\rt \xbsquig B$, $\xbsquig E \xbsquig A _\rt \xbsquig C$, $\xbsquig E \xbsquig A _\rt \xbsquig D$ and $\xbsquig E \xbsquig A _\rt \xbsquig E$ are all empty.
\end{itemize}
This completes the proof.
\end{proof}

\noindent
With this in hand, we can show termination.
\begin{restatable}{proposition}{termination}
\label{termination}
Standard reduction is terminating on valid types, terms and substitutions.
\end{restatable}
\begin{proof}
Standard reduction on types and substitutions is given in terms of standard reduction of a finite family of terms, so we need only check that standard reduction of terms has no infinite sequences. The variable case is trivial, so we consider reduction of some coherence term $t \equiv \coh (\Gamma:T)[\sigma]$.

We proceed by simultaneous induction, on the dimension of $t$, and on subterms of~$t$.
Since a subterm of a valid term can never have a higher dimension, this is well-defined. The dimension is not defined for variables, but since variables are in normal form, this does not affect the argument. We also note that any subterm of a valid term is itself valid.

 Thanks to Proposition~\ref{prop:progressive}, we know that the standard reduction of a given term can be separated into finitely many distinct phases of standard A-, B-, C-, D- and E-reduction, with each phase involving standard reductions of a single fixed type. So we need only show that each standard reduction phase will terminate.
\begin{itemize}
\item \emph{Standard A-reduction.}
For the first phase of A-reductions, we can use an induction on subterms, since a substitution is of finite length, a given term will have only finitely many standard A\-reductions. For the second phase of A-reductions, we have that these reductions followed an E-reduction, and so all the terms in the substitution have a lower dimension than the original term, and so we can induct on dimension and apply the same reasoning as before.
\item \emph{Standard B-reduction.} Since the context has a finite number of variables, a given term will have only finitely many standard B-reductions.
\item \emph{Standard C-reduction.} Since the cell part of a valid term is formed from terms of strictly lower dimension, and is itself valid in the pasting context of the head coherence, it follows by induction on dimension that a term will have only finitely many standard C-reductions.
\end{itemize}
Standard D- and E-reductions are single-step operations, so no termination argument is necessary for those.
\end{proof}

\noindent
Having now established that standard reduction has no infinite sequences, it is clear every term $t$ has a unique normal form, \(N(t)\), to which it reduces.

\subsection{Technical results on reduction}

Here we collect further results on reduction, mostly of a technical nature, which will be used in the next subsection for our main proof.

\begin{restatable}{lemma}{reductionapplication}
\label{reductionapplication}
Term substitution is compatible with substitution reduction:
$$\sigma \squig \sigma' \quad \Rightarrow \quad u \sub \sigma \squig_{\rt} u \sub {\sigma'}$$
\end{restatable}
\begin{proof}
We induct on the structure of $u$.  If $u$ is a variable, then either $u\sub \sigma \equiv \sigma' \sub \tau$ or $u\sub \sigma \squig u\sub {\sigma'}$; in either case, we have $u\sub \sigma \squig_{\r} u \sub {\sigma'}$ as required.

Otherwise, we have $u \equiv \coh (\Gamma: U)[\rho]$, and we argue as follows:
\begin{align*}
u[\sigma] \equiv{}& \coh (\Gamma : U)[\rho \circ \sigma]
\\
\equiv{}& \coh(\Gamma : U)[ \rho_1 \sub \sigma, \rho_2\sub \sigma, \ldots, \rho_n\sub \sigma ]
\\
\xsquig A _\rt {}& \coh(\Gamma : U)[\rho_1\sub {\sigma'}, \rho_2\sub \sigma, \ldots, \rho_n\sub \sigma ]
\\
\xsquig A _\rt{} & \coh(\Gamma : U)[ \rho_1\sub {\sigma'}, \rho_2\sub {\sigma'}, \ldots, \rho_n\sub {\sigma} ]
\\
\xsquig A _\rt{} & \cdots
\\
\xsquig A _\rt{} & \coh(\Gamma : U)[ \rho_1\sub {\sigma'}, \rho_2\sub {\sigma'}, \ldots, \rho_n\sub {\sigma'} ]
\\
\equiv{}& \coh(\Gamma : U)[\rho \circ \sigma']
\\
\equiv{}& u \sub {\sigma'}
\end{align*}
Hence $u\sub{\sigma} \squig_\rt u\sub{\sigma'}$ as required.
\end{proof}

\begin{restatable}{lemma}{reducelefttermsub}
\label{reducelefttermsub}
Term substitution is compatible with term reduction:
$$u \squig u' \quad\Rightarrow\quad u\sub{\sigma} \squig u'\sub{\sigma}$$
\end{restatable}
\begin{proof}
If $u$ is a variable, no reduction is possible, in contradiction with hypothesis. We therefore assume $u \equiv \coh (\Gamma : T)[\mu]$ is a coherence term, writing $\mu = [\mu_1, \ldots, \mu_n]$, and proceed by case analysis on the structure of the reduction $u \squig u'$, and by induction on subterms of $u$.

If $u \xsquig A u'$ via some $m_i \squig m_i'$, then we argue as follows:
\begin{align*}
u \sub \sigma &\equiv \coh(\Gamma:U)[\mu \circ \sigma]
\\
&\equiv \coh(\Gamma : U)[ \mu_1\sub{\sigma}, \ldots, \mu_i\sub{\sigma}, \ldots, \mu_k\sub{\sigma} ]
\\
& \xsquig A \coh(\Gamma : U)[\mu_1\sub{\sigma}, \ldots, \mu_i'\sub{\sigma}, \ldots, \mu_k\sub{\sigma} ]
\\
& \equiv u \sub {\sigma'}
\intertext{Alternatively, we suppose $u \xsquig B u'$ is the B-reduction $\coh(\Gamma : U)[\mu] \xsquig B \coh(\Gamma \sslash x : U\sub{\pi_x})[\mu \sslash x]$, eliminating some locally-maximal variable $x_i$ of $\Gamma$ for which $x_i \sub \mu \equiv \mu_i$ is an identity. Then $\mu_i\sub{\sigma}$ is also an identity, and hence using Lemma~\ref{lem:quotientsub} we have:}
u \sub \sigma &\equiv \coh (\Gamma : U) [\mu \circ \sigma]
\\
&\equiv\coh (\Gamma : U) [\mu \circ \sigma]
\\
&\xsquig B \coh(\Gamma \sslash x: U \sub {\pi_x} )[(\mu \circ \sigma) \sslash x]
\\
&\equiv \coh (\Gamma \sslash x:U \sub{\pi_x} )[(\mu \sslash x) \circ \sigma]
\\
&\equiv u' \sub \sigma
\intertext{If $u \xsquig C u'$ via some $T \xsquig C T'$, then we argue as follows:}
u \sub\sigma &\equiv \coh(\Gamma : T)[\mu \circ \sigma] \xsquig C \coh (\Gamma : T')[\mu \circ \sigma] \equiv u'\sub{\sigma}
\intertext{If $u \xsquig D u'$ as $\coh (\uD^n:S^{n-1})[\ldots , u'] \xsquig D u'$, then we argue as follows:}
u \sub \sigma &\equiv \coh(\uD^n : S^{n-1})[\ldots, u' \sub\sigma] \xsquig D u' \sub \sigma
\intertext{Finally, if $u \xsquig E u'$ as $\coh (\Gamma:u \to _U u)[\mu] \xsquig E \i_{\dim U+1} \sub{\{U,u\} \circ \mu}$, then we argue as follows, using Lemma~\ref{lem:typeunfoldingsub}:}
u \sub \sigma &\equiv \coh (\Gamma : u \to_U u)[\mu \circ \sigma]
\\
&\xsquig E \i_{\dim U  + 1} \sub{\{U,u\} \circ \mu \circ \sigma}
\\
&\equiv \i_{\dim U + 1} \sub{\{U,u\} \circ \mu} \sub{\sigma}
\\
&\equiv u' \sub \sigma
\end{align*}
This completes the proof.
\end{proof}

\begin{restatable}[Identities reduce to identities]{lemma}{preserveidentities}
\label{lem:preserveidentities}
If $u$ is an identity, and $u \squig u'$, then $u'$ is an identity.
\end{restatable}
\begin{proof}
We recognize an identity term by looking at the head. We prove the result by case analysis on the reduction $u \squig u'$. If $u \xsquig A u'$ the result is immediate, since A-reductions do not change the head. If $u \xsquig B u'$ we have a contradiction, since identity terms cannot be B-redexes by definition. If $u \xsquig C u'$ we again have a contradiction, since a C-reduction acts on the head of the term, but the head of an identity term is in normal form. If $u \xsquig D u'$ we again have a contradiction, since identity terms have the wrong form to be D-redexes. If $u \xsquig E u'$ we again have a contradiction, since identity terms cannot be E-redexes by definition.
\end{proof}

\begin{restatable}{lemma}{standardslash}
\label{standardslash}
If $\sigma \bsquig \tilde \sigma$, then $\sigma \sslash x \bsquig_\r \tilde \sigma \sslash x$.
\end{restatable}
\begin{proof}
This is immediate, since $\sigma \sslash x$ is a sublist of $\sigma$. If the first reducible argument of $\sigma$ is not in the sublist $\sigma \sslash x$, then $\sigma \sslash x \equiv \tilde \sigma \sslash x$. Otherwise, it will still be the first reducible argument of the sublist, and  $\sigma \sslash x \bsquig \tilde \sigma \sslash x$.
\end{proof}

\begin{restatable}{lemma}{prunecommute}
\label{lem:prunecommute}
Given a context $\Gamma$ and distinct locally-maximal variables $x,y$, and a substitution $\Delta \vdash \sigma : \Gamma$ with $x\sub{\sigma},y\sub{\sigma}$ both identities, the following contexts and substitutions are identical:
\begin{mathpar}
(\Gamma \sslash x) \sslash y \equiv (\Gamma \sslash y) \sslash x
\and
\pi_x \circ \pi_y \equiv \pi_y \circ \pi_x
\and
(\sigma \sslash x) \sslash y \equiv (\sigma \sslash y) \sslash x
\end{mathpar}
\end{restatable}
\begin{proof}
The first statement is clear from example, in this case showing $(\Gamma \sslash \Omega) \sslash \delta \equiv (\Gamma \sslash \delta) \sslash \Omega$, since the excision operations are independent:
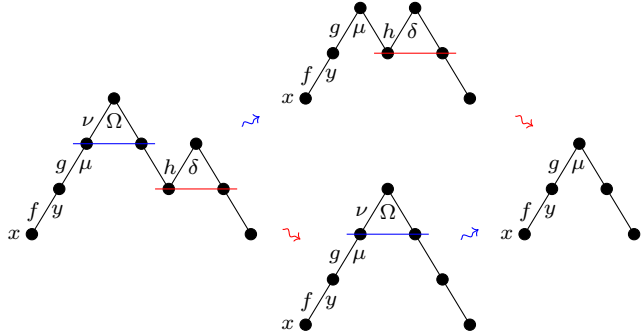
\[
\tikzset{point/.style={draw, circle, black, fill, inner sep=1.5pt, label={}}}
\begin{aligned}
\begin{tikzpicture}[xscale=.6, scale=.6, font=\scriptsize, every node/.style={inner sep=2pt}]
\draw (0,0) node [point] {} node [left=2pt] {$x$} to node [left] {$f$} node [right] {$y$} (1,1) node [point] {} to node [left] {$g$} node [right] {$\mu$} (2,2) node [point] {} to node [left] {$\nu$} node [right] {$\Omega$} (3,3) node [point] {} to (4,2) node [point] {} to (5,1) node [point] {} to node [left] {$h$} node [right] {$\delta$} (6,2) node [point] {} to (7,1) node [point] {} to (8,0) node [point] {};
\draw [blue] (1.5,2) to +(3,0);
\draw [red] (4.5,1) to +(3,0);

\begin{scope}[xshift=10cm, yshift=3cm]
\draw (0,0) node [point] {} node [left=2pt] {$x$} to node [left] {$f$} node [right] {$y$} (1,1) node [point] {} to node [left] {$g$} node [right] {$\mu$} (2,2) node [point] {} to (3,1) node [point] {} to node [left] {$h$} node [right] {$\delta$} (4,2) node [point] {} to (5,1) node [point] {} to (6,0) node [point] {};
\draw [red] (2.5,1) to +(3,0);
\end{scope}

\begin{scope}[xshift=10cm, yshift=-2cm]
\draw (0,0) node [point] {} node [left=2pt] {$x$} to node [left] {$f$} node [right] {$y$} (1,1) node [point] {} to node [left] {$g$} node [right] {$\mu$} (2,2) node [point] {} to node [left] {$\nu$} node [right] {$\Omega$} (3,3) node [point] {} to (4,2) node [point] {} to (5,1) node [point] {} to (6,0) node [point] {};
\draw [blue] (1.5,2) to +(3,0);
\end{scope}

\begin{scope}[xshift=18cm, yshift=0cm]
\draw (0,0) node [point] {} node [left=2pt] {$x$} to node [left] {$f$} node [right] {$y$} (1,1) node [point] {} to node [left] {$g$} node [right] {$\mu$} (2,2) node [point] {} to (3,1) node [point] {} to (4,0) node [point] {};
\end{scope}

\node [red, rotate=-30] at (9.5,0) {$\squig$};
\node [blue, rotate=30] at (8,2.5) {$\squig$};
\node [blue, rotate=30] at (16,0) {$\squig$};
\node [red, rotate=-30] at (18,2.5) {$\squig$};

\end{tikzpicture}
\end{aligned}
\]
The other claims follows similarly.
\end{proof}

\begin{restatable}{lemma}{discsubtrick}
\label{lem:discsubtrick}
For a valid substitution $\Delta \vdash \sigma : \uD^n$, we have for all $i \leq n$:
\begin{align*}
d_i \sub \sigma &= \src^{n-i}(d_n \sub \sigma)
&
d'_i \sub \sigma &= \tgt^{n-i}(d_n \sub \sigma)
\end{align*}
\end{restatable}

\begin{proof}
Extending the substitution construction rule along definitional equality we must have $\src(d_n \sub \sigma) = \src(d_n) \sub \sigma$, and similarly for $\tgt$. The result follows.
\end{proof}

\subsection{Standard reduction generates definitional equality}

We will show that the symmetric, transitive, reflexive closure of standard reduction generates definitional equality. Since we have already shown that standard reduction is a terminating reduction strategy, this gives an algorithm to determine whether two given terms are definitionally equal, by computing their standard normal forms and checking if they are syntactically equal (up to $\alpha$\-equivalence.)

\begin{restatable}{lemma}{closerts}
\label{close_rts}
If $a \bsquig_\rts b$,  $\exists \, c$ with $a \bsquig_\rt c$ and $b \bsquig_\rt c$.
\end{restatable}

\begin{proof}
This is immediate, since $\bsquig$ is a reduction strategy.
\end{proof}

\noindent The next proposition tells us that standard reduction agrees with equality.

\begin{restatable}{proposition}{standardise}
\label{standardise}
For all terms $s$, such that \(s\) is valid in some context \(\Theta\), we have the following:
\begin{enumerate}
\item[(i)] If $s \squig t$, we can find terms $a, b$ which admit a  reduction $a \squig_\r b$ and standard reductions $s \bsquig_\t a$, $t \bsquig_\rts b$, illustrated as follows:
\[
\begin{tikzpicture}[xscale=1.2, yscale=1.06]
\node (s) at (0,0) {$s$};
\node (t) at (1,0) {$t$};
\node (a) at (0,-1) {$a$};
\node (b) at (1,-1) {$b$};
\draw [nonstandard] (s) to (t);
\draw [nonstandard] (a) to node [below] {\r} (b);
\draw [standard] (s) to node [left] {\t} (a);
\draw [standard] (t) to node [right] {\rts} (b);
\end{tikzpicture}
\]
\item[(ii)] If $\Theta \vdash s = t$, then $s \bsquig_\rts t$.
\end{enumerate}
\end{restatable}

\begin{proof}
Since by Proposition~\ref{termination} every term reaches normal form after a finite number of standard reductions, and since by Lemma~\ref{equalityreduction} reduction generates equality, it is clear that statement (i) implies statement (ii). We therefore focus here on the proof of (i).

We can neglect the case of $s$ being a variable, since variables do not reduce. It follows that $s$ is a coherence term, and since the dimension of a valid coherence term is well-defined, we will make use of that throughout. The proof of (i) is by simultaneous induction  on the dimension of $s$, and on subterms of $s$. Since no subterm of $s$ has a greater dimension than $s$ itself, this is consistent. We also note that if we have property (i) for all terms of dimension less than \(n\), then we can derive property (ii) for all terms of dimension less than \(n\), and so we may assume (ii) holds on terms of lower dimension by induction on dimension.

One possibility, which arises several times in the case analysis below, is that the reduction $s \squig t$ is itself standard (that is, $t \equiv \tilde s$.) We can handle this case once-and-for-all as follows:
\begin{equation}
\label{alreadystandard}
\begin{aligned}
\begin{tikzpicture}[xscale=1.2, yscale=1.06]
\node (s) at (0,0) {$s$};
\node (t) at (1,0) {$t$};
\node (a) at (0,-1) {$\tilde s$};
\node (b) at (1,-1) {$\tilde s$};
\draw [nonstandard] (s) to (t);
\draw [triple] (a) to node [below] {} (b);
\draw [standard] (s) to node [left] {} (a);
\draw [triple] (t) to node [right] {} (b);
\end{tikzpicture}
\end{aligned}
\end{equation}
We refer to this argument below where it is needed.

We now begin the main proof of property (i), by case analysis on the reduction $s \squig t$.

\proofstep \emph{First case $s \xsquig {A} t$.}
We suppose $s \equiv \coh (\Gamma : T)[\sigma]$, and that $s \xsquig A t$ by reducing an argument of $\sigma$ via $s_i^{} \squig s_i'$. Then $s$ must be a standard A-redex, because if some argument is not in normal form, there must exist a leftmost argument not in normal form; so we have $s \xbsquig A \tilde s$.

Suppose $s \xsquig A t$ and $s \xbsquig A \tilde s$ act by reducing the same argument of $s_i$ of $\sigma$, via $s_i \squig s_i'$ and $s_i \bsquig \tilde s_i$ respectively. If $s'_i \equiv \tilde s_i^{}$ then we are done by the argument above expression~\eqref{alreadystandard} above. Otherwise the result holds by induction on the subterm $s_i$, as follows:
$$
\begin{tikzpicture}[xscale=4.5, yscale=1.2]
\node (tl) at (0,0) {$\coh (\Gamma : T)[ \ldots, s_i, \ldots ]$};
\node (tr) at (1,0) {$\coh (\Gamma : T)[ \ldots, s_i', \ldots ]$};
\node (bl) at (0,-1) {$\coh (\Gamma : T)[ \ldots, \tilde s_i, \ldots ]$};
\node (br) at (1,-1) {$\coh (\Gamma : T)[ \ldots, q, \ldots ]$};
\draw [nonstandard] (tl) to node [above] {A} (tr);
\draw [standard] (tl) to node [left] {A} (bl);
\draw [standard] (tr) to node [left] {\rts} node [right] {A} (br);
\draw [nonstandard] (bl) to node [below] {A} (br);
\end{tikzpicture}
$$
Alternatively, suppose they act by reducing different arguments of $\sigma$. Then the redexes are independent, and we argue as follows:
$$
\begin{tikzpicture}[xscale=6, yscale=1.5,font=\small]
\node[align=left](tl) at (0,0) {$\coh(\Gamma : T)[ \ldots, s_i, \ldots, s_j, \ldots ]$};
\node[align=left](tr) at (1,0) {$\coh (\Gamma : T)[ \ldots, s_i^{}, \ldots, s_j', \ldots ]$};
\node[align=left](bl) at (0,-1) {$\coh (\Gamma : T)[ \ldots, \tilde s_i, \ldots, s_j, \ldots]$};
\node[align=left](br) at (1,-1) {$\coh (\Gamma : T)[ \ldots, \tilde s^{}_i, \ldots, s_j', \ldots]$};
\draw [nonstandard] (tl) to node [above] {A} (tr);
\draw [standard] (tl) to node [left] {A} (bl);
\draw [standard] (tr) to node [left] {} node [right] {A} (br);
\draw [nonstandard] (bl) to node [below] {A} (br);
\end{tikzpicture}
$$

\proofstep \emph{Second case $s \xsquig {B} t$.}
For this case, we suppose $s \xsquig{B} t$ as follows:
$$s \equiv \coh (\Gamma: T)[\sigma] \xsquig B \coh (\Gamma \sslash x: T \sub{\pi_x} )[\sigma \sslash x] \equiv t$$
We proceed by case analysis on the standard reduction $s \bsquig \tilde s$.

\smallskip
\noindent \emph{Standard A-reduction.} In this case, we argue as follows:
$$
\begin{tikzpicture}[xscale=4, yscale=1.2]
\node (tl) at (0,0) {$\coh (\Gamma: T)[\sigma]$};
\node (tr) at (1,0) {$\coh (\Gamma \sslash x: T \sub{\pi_x})[\sigma \sslash x]$};
\node (bl) at (0,-1) {$\coh (\Gamma : T)[\tilde \sigma]$};
\node (br) at (1,-1) {$\coh (\Gamma \sslash x: T \sub{\pi_x})[\tilde \sigma \sslash x]$};
\draw [nonstandard] (tl) to node [above] {B} (tr);
\draw [standard] (tl) to node [left] {A} node [right] {} (bl);
\draw [standard] (tr) to node [right] {A} (br);
\draw [nonstandard] (bl) to node [below] {B} (br);
\end{tikzpicture}
$$
Since the upper B-reduction is valid, we know that $x\sub{\sigma}$ is an identity; then by Lemma~\ref{lem:preserveidentities} we also have that $x\sub{\tilde \sigma}$ is an identity, and so the lower B-reduction is valid. Validity of the  standard A-reduction on the right of the square follows from Lemma~\ref{standardslash}.

\smallskip
\noindent \emph{Standard B-reduction.} We suppose  $s$ is a standard B-redex with respect to some locally-maximal variable $y$. If $x \equiv y$, then the reductions are the same, and this case is handled by the general argument above expression~\eqref{alreadystandard}. Otherwise, we argue as follows, using Lemma~\ref{lem:prunecommute}:
$$
\begin{tikzpicture}[xscale=4.5, yscale=1.4,font=\footnotesize]
\node (tl) at (0,0) {$\coh (\Gamma: T)[\sigma]$};
\node (tr) at (1,0) {$\coh (\Gamma \sslash x: T \sub{\pi_x})[\sigma \sslash x]$};
\node[align=left] (bl) at (0,-2) {$\coh (\Gamma \sslash y: T \sub{\pi_y})$\\\hspace{.7cm}$[\sigma \sslash y]$};
\node[align=left] (br) at (1,-2) {$\coh \big((\Gamma \sslash y) \sslash x : T \sub{\pi_y} \sub{\pi_x} \big)[(\sigma \sslash y) \sslash x]$};
\node[align=left] (rr) at (1,-1) {$\coh \big((\Gamma \sslash x) \sslash y : T \sub{\pi_x} \sub{\pi_y} \big)[(\sigma \sslash x) \sslash y]$};
\draw [nonstandard] (tl) to node [above] {B} (tr);
\draw [standard] (tl) to node [left] {B} node [right] {} (bl);
\draw [standard] (tr) to node [right] {B} (rr);
\draw [nonstandard] (bl) to node [below] {B} (br);
\draw [triple] (rr) to (br);
\end{tikzpicture}
$$

\proofstep \emph{Third case $s \xsquig {C} t$.}
Suppose $s \equiv \coh (\Gamma : T)[\sigma]$, and $s \xsquig{C} t$ acts via some type reduction $T \squig U$. Then we argue by case analysis on the reduction $s \bsquig \tilde s$.

\smallskip
\noindent \emph{Standard A-reduction.} In this case we have the following:
$$
\begin{tikzpicture}[xscale=4, yscale=1.2]
\node (tl) at (0,0) {$\coh (\Gamma: T)[\sigma]$};
\node (tr) at (1,0) {$\coh (\Gamma: U)[\sigma]$};
\node (bl) at (0,-1) {$\coh (\Gamma : T)[\tilde \sigma]$};
\node (br) at (1,-1) {$\coh (\Gamma: U)[\tilde \sigma]$};
\draw [nonstandard] (tl) to node [above] {C} (tr);
\draw [standard] (tl) to node [left] {A} (bl);
\draw [standard] (tr) to node [right] {A} (br);
\draw [nonstandard] (bl) to node [below] {C} (br);
\end{tikzpicture}
$$

\smallskip
\noindent \emph{Standard B-reduction.} If $s$ is a standard B-redex, we have the following, employing Lemma~\ref{reducelefttermsub}:
$$
\begin{tikzpicture}[xscale=4.5, yscale=1.2]
\node (tl) at (0,0) {$\coh (\Gamma: T)[\sigma]$};
\node (tr) at (1,0) {$\coh (\Gamma: U)[\sigma]$};
\node (bl) at (0,-1) {$\coh (\Gamma \sslash x : T \sub{\pi_x})[ \sigma \sslash x ]$};
\node (br) at (1,-1) {$\coh (\Gamma \sslash x: U \sub{\pi_x} )[\sigma \sslash x]$};
\draw [nonstandard] (tl) to node [above] {C} (tr);
\draw [standard] (tl) to node [left] {B} (bl);
\draw [standard] (tr) to node [right] {B} (br);
\draw [nonstandard] (bl) to node [below] {C} (br);
\end{tikzpicture}
$$

\smallskip
\noindent \emph{Standard C-reduction.} Suppose $s$ is a standard C-redex via a standard type reduction $T \bsquig \tilde T$. If $\tilde T \equiv U$, we are done by the argument above expression~\eqref{alreadystandard}. Otherwise, since $T \squig U$, by induction on subterms we know that $U \bsquig_\rts T$, and we argue as follows:
$$
\begin{tikzpicture}[xscale=4, yscale=2]
\node (tl) at (0,0) {$\coh (\Gamma: T)[\sigma]$};
\node (tr) at (1,0) {$\coh (\Gamma: U)[\sigma]$};
\node (bl) at (0,-1) {$\coh (\Gamma : \tilde T)[ \tau ]$};
\node (br) at (1,-1) {$\coh (\Gamma: \tilde T)[\tau]$};
\node (brr) at (1,-.5) {$\coh (\Gamma: T)[\tau]$};
\draw [nonstandard] (tl) to node [above] {C} (tr);
\draw [standard] (tl) to node [left] {C} node [right] {\t} (bl);
\draw [standard] (tr) to node [right] {C} node [left] {\rts} (brr);
\draw [standard] (brr) to node [right] {C} (br);
\draw [triple] (bl) to node [below] {} (br);
\end{tikzpicture}
$$

\proofstep \emph{Fourth case $s \stackrel {\mathrm{D}} \squig t$.}
We suppose $s \stackrel {\mathrm{D}} \squig t$ as follows:
$$s \equiv \coh (\uD^{n+1}: S^n)[ \ldots, t] \xsquig D t$$
Then we consider the standard reduction for $s$.

\smallskip
\noindent\emph{Standard A-reduction.}
Supposing $t \bsquig \tilde t$ is the leftmost reducible argument of $s$, then we have the following:
$$
\begin{tikzpicture}[xscale=4, yscale=1.2]
\node (tl) at (0,0) {$\coh (\uD^{n+1}: S^n)[ \ldots, t ]$};
\node (tr) at (1,0) {$t$};
\node (bl) at (0,-1) {$\coh (\uD^{n+1}: S^n)[ \ldots, \tilde t]$};
\node (br) at (1,-1) {$\tilde t$};
\draw [nonstandard] (tl) to node [above] {D} (tr);
\draw [standard] (tl) to node [left] {A} (bl);
\draw [standard] (tr) to node [left] {} (br);
\draw [nonstandard] (bl) to node [below] {D} (br);
\end{tikzpicture}
$$
Otherwise, let $p$ be the leftmost reducible argument of $s$. Then we argue as follows:
$$
\begin{tikzpicture}[xscale=4.5, yscale=1.2]
\node (tl) at (0,0) {$\coh (\uD^{n+1}: S^n)[ \ldots, p, \ldots, t]$};
\node (tr) at (1,0) {$t$};
\node (bl) at (0,-1) {$\coh (\uD^{n+1}: S^n)[ \ldots, \tilde p, \ldots, t]$};
\node (br) at (1,-1) {$t$};
\draw [nonstandard] (tl) to node [above] {D} (tr);
\draw [standard] (tl) to node [left] {A} (bl);
\draw [triple] (tr) to node [left] {} (br);
\draw [nonstandard] (bl) to node [below] {D} (br);
\end{tikzpicture}
$$

\smallskip
\noindent \emph{Standard B-reduction.}
We must have $t \equiv \i_n \sub{\ldots, q_2, q_1 }$, and hence $s \equiv \coh (\uD^{n+1}: S^n)[\ldots, p_2^{}, p_1^{}, p_1', \i_n \sub{\ldots, q_2, q_1 } ]$. Since $s$ is valid we deduce $q_1^{} = p_1^{} = p'_1$ and $q_i = p_i$. It follows by induction on dimension that $q_i \bsquig_\rts p_i$. We put this together as follows:
$$
\begin{tikzpicture}[xscale=5, yscale=1.2,font=\footnotesize]
\node[align=left] (tl) at (0,0) {$\coh (\uD^{n+1}: S^n)[\ldots, p_2^{}, p_1^{}, p_1', \i_n \sub{\ldots, q_2, q_1 } ]$};
\node (tr) at (1,0) {$\i_n \sub{\ldots, q_2, q_1 }$};
\node (bl) at (0,-1) {$\coh (\uD^{n}: d_n \to_{S^{n-1}} d_n)[ \ldots, p_2, p_1 ]$};
\node (br) at (1,-1) {$\i_n \sub{\ldots, p_2, p_1 }$};
\draw [nonstandard] (tl) to node [above] {D} (tr);
\draw [standard] (tl) to node [left] {B} (bl);
\draw [standard] (tr) to node [right] {A} node [left] {\rts} (br);
\draw [triple] (bl) to (br);
\end{tikzpicture}
$$

\smallskip
\noindent {\emph{Standard C-reduction}.} The term $s$ cannot be a  C-redex, since the type $S^n$ is in normal form, being constructed entirely from variables.

\smallskip
\noindent {\emph{Standard D-reduction}.}  In this case, the result follows from the argument around around expression~\eqref{alreadystandard} above.

\proofstep \emph{Fifth case $s \stackrel {\mathrm{E}} \squig t$.}
We suppose $s \stackrel {\mathrm{E}} \squig t$ as follows, for $n = \dim \, T + 1$:
$$s \equiv \coh (\Gamma: u \to_T u)[\sigma] \squig \i_{n} \sub{\{T,u\} \circ \sigma} \equiv t$$
We now consider the structure of the standard reduction $s \bsquig \tilde s$.

\smallskip
\noindent \emph{Standard A-reduction.} If $\sigma \bsquig \tilde \sigma$ then we get by Lemma~\ref{lem:eq-sub-r-inv} that \(\{T,u\}\circ\sigma = \{T,u\}\circ\tilde\sigma\). We conclude by induction on dimension that $\i_{n} \sub{\{T,u\} \circ \sigma} \xbsquig A _\rts \i_{n} \sub{\{T,u\} \circ \tilde\sigma}$. Altogether, we have the following as required:
$$
\begin{tikzpicture}[xscale=4, yscale=1.2]
\node (tl) at (0,0) {$\coh (\Gamma : u \to_T u)[ \sigma]$};
\node (tr) at (1,0) {$\i_{n} \sub{\{T,u\} \circ \sigma}$};
\node (bl) at (0,-1) {$\coh (\Gamma : u \to_T u)[\tilde \sigma]$};
\node (br) at (1,-1) {$\i_{n} \sub{\{T,u\} \circ \tilde \sigma}$};
\draw [nonstandard] (tl) to node [above] {E} (tr);
\draw [standard] (tl) to node [left] {A} (bl);
\draw [standard] (tr) to node [left] {\rts} node [right] {A} (br);
\draw [nonstandard] (bl) to node [below] {E} (br);
\end{tikzpicture}
$$

\smallskip
\noindent \emph{Standard B-reduction.} In this case we have the following:
$$
\begin{tikzpicture}[xscale=5, yscale=.7, font=\footnotesize]
\node (tl) at (0,0) {$\coh (\Gamma : u \to_T u)[ \sigma]$};
\node (tr) at (1,0) {$\i_{n} \sub{\{T,u\} \circ \sigma}$};
\node [scale=1] (bl) at (0,-2) {$\coh (\Gamma \sslash x: (u \!\to\!_T u) \sub{\pi_x})[\sigma \sslash x]$};
\node (br) at (1,-2) {$\i_n \sub{\{T\sub{\pi_x},u\sub{\pi_x}\} \!\circ\! \sigma \sslash x}$};
\draw [nonstandard] (tl.east) to node [above] {E} (tr.west);
\draw [standard] (tl) to node [left] {B} (bl);
\draw [standard] (tr) to node [left] {$\rts$} node [right] {A} (br);
\draw [nonstandard] (bl.east) to node [below] {E} (br.west);
\end{tikzpicture}
$$
We obtain the right-hand standard A-reduction as follows. From Lemma~\ref{lem:pruning-construct} we know $\sigma = \pi_x \circ (\sigma \sslash x)$, and hence:
\[\{T\sub{\pi_x},u\sub{\pi_x}\} \circ \sigma \sslash x \equiv \{T,u\} \circ \pi_x \circ \sigma \sslash x = \{T,u\} \circ \sigma\]
By induction on dimension we conclude $\{T,u\}\circ \sigma \bsquig_\rts \{T\sub{\pi_x},u\sub{\pi_x}\} \circ \sigma \sslash x$.

\smallskip
\noindent \emph{Standard C-reduction.}
Since $\{T,u\} = \{N(T),N(u)\}$ and so \(\{T,u\} \circ \sigma = \{N(T),N(u)\} \circ \sigma\), it follows by induction on dimension that $\{T,u\} \circ \sigma \bsquig_\rts \{N(T),N(u)\} \circ\sigma$, and hence that $\i_{n} \sub{\{T,u\} \circ \sigma} \xbsquig A  _\rts \i_{n} \sub{\{N(T),N(u)\} \circ \sigma}$. We then have the following:
$$
\begin{tikzpicture}[xscale=4.2, yscale=1.5,font=\small]
\node (tl) at (0,0) {$\coh (\Gamma : u \to_T u)[ \sigma]$};
\node (tr) at (1,0) {$\i_{n} \sub{\{T,u\} \circ \sigma}$};
\node (br) at (1,-1) {$\i_{n} \sub{\{N(T),N(u)\} \circ \sigma}$};
\node (bbbl) at (0,-1) {$\coh (\Gamma : N(u) \to_{N(T)} N(u))[\sigma]$};
\draw [nonstandard] (tl) to node [above] {E} (tr);
\draw [standard] (tr) to node [left] {\rts} node [right] {A} (br);
\draw [nonstandard] (bbbl) to node [below] {E} (br);
\draw [standard] (tl) to node [left] {C} node [right] {\t} (bbbl);
\end{tikzpicture}
$$

\smallskip
\noindent \emph{Standard D-reduction.} The term $s$ cannot be a D-redex, as the term $s$ has an incompatible structure.

\smallskip
\noindent \emph{Standard E-reduction.} In this case we are done, with the result following from the argument around around expression~\eqref{alreadystandard}.

\noindent
This completes the argument.
\end{proof}

Finally, we can use these results to show that the equality relation is decidable.

\begin{restatable}{theorem}{decidedefeq}
\label{cor:decidedefeq}
We have $s = t$ if and only if $N(s) \equiv N(t)$.
\end{restatable}

\begin{proof}
  If \(N(s) \equiv N(t)\) then \(s \bsquig_\rts t\) and so \(s = t\). Conversely, if \(s = t\), then by Proposition~\ref{standardise}, we have that \(s \bsquig_\rts t\), and so by Lemma~\ref{close_rts} we have \(s \bsquig_\rt u\) and \(t \bsquig_\rt u\) and so \(N(s) \equiv N(u) \equiv N(t)\).
\end{proof}

\section{Examples and Implementation}

In this section, we investigate some examples, to see the type theory in action. As we can decide equality by comparing normal forms, we are able to provide an algorithm for type-checking \(\Catt_\su\) terms. We have implemented type theory in OCaml, and made it available here:
$$
\text{\url{http://github.com/ericfinster/catt.io/tree/v0.1}}
$$
The following examples demonstrate the use of this implementation, and the code for all of them can be found in the examples folder within the repository.

\begin{example}[Left unitor] \

\vspace{0pt}
\noindent
{\scriptsize\verb|examples/example_left_unitor.catt|}

\vspace{2pt}
\noindent
The left unitor on \(f : x \to y\) is defined by the following \(\Catt\) term, of type $\comp_2 \sub{\i_0\sub x, f} \to f$:
  \begin{align*}
    \coh((x : \star)(y : \star)(f : x \to_\star y) : \mathsf{comp_2}\sub{x,x,\i_0\sub{x},y,f} \to f)[x,y,f]
  \end{align*}
  In \(\Catt_\su\) this reduces to the identity on \(f\), demonstrating how higher coherence data can trivialise in the theory.
\end{example}

\begin{example}[Associator on a unit]\

\vspace{0pt}
\noindent
{\scriptsize\verb|examples/example_associator_unit.catt|}

\vspace{3pt}
\noindent
For \(f : x \to y\) and \(g : y \to z\), we have the following defin-itional equality in $\Catt_\su$:
$$(f \cdot \id_y) \cdot g = f \cdot (\id_y \cdot g)$$
We may therefore anticipate the associator $\smash{\alpha_{f,\id_y,g}}$ to trivialise, and indeed this is the case. We illustrate this with the following standard reduction sequence, where for brevity we leave some substitution arguments implicit:
{
    \begin{alignat*}{2}
      &\phantom{{}\equiv{}}&&\mathsf{assoc}\sub{f,\i_0\sub{y},g}\\
      &{}\equiv{}&& \coh((x : \star) (y : \star) (a : x \to y) (z : \star) (b : y \to z) (w : \star) (c : z \to w) :\\[-2pt]
      &&&\quad\mathsf{comp_2}\sub{\mathsf{comp_2}\sub{a,b},c} \to \mathsf{comp_2}\sub{a,\mathsf{comp_2}\sub{b,c}})\sub{f,\i_0\sub{y},g}\\
      &{}\xbsquig {B}{}&& \coh((x : \star) (y : \star) (a : x \to y) (w : \star) (c : y \to w):\\[-2pt]
      &&&\quad\mathsf{comp_2}\sub{\mathsf{comp_2}\sub{a,\i_0\sub{y}}, c} \to \mathsf{comp_2} \sub{a, \mathsf{comp_2}\sub{\i_0\sub{y}, c}})\sub{f,g}\\
      &{}\xbsquig {C}_\t{}&& \coh((x : \star) (y : \star) (a : x \to y) (w : \star) (c : y \to w):\\[-2pt]
      &&&\quad\mathsf{comp_2}\sub{\mathsf{comp_1}\sub{a}, c} \to \mathsf{comp_2} \sub{a, \mathsf{comp_1}\sub{c}})\sub{f,g}\\
      &{}\xbsquig {C}_\t{}&& \coh((x : \star) (y : \star) (a : x \to y) (w : \star) (c : y \to w):\\[-2pt]
      &&&\quad\mathsf{comp_2}\sub{a, c} \to \mathsf{comp_2} \sub{a,c})\sub{f,g}\\
      &{}\xbsquig {E}&&\i_1\sub{\{(x \to z), \mathsf{comp_2}\sub{f,g}\}}
    \end{alignat*}}%
This reduction sequence uses all three generators \textsc{prune}, \textsc{disc} and \textsc{endo}, with disc removal being used for the second  C-reduction.
\end{example}

\begin{example}[Eckmann-Hilton move]
\label{ex:eh}
\

\vspace{0pt}
\noindent
{\scriptsize\verb|examples/example_eh_catt.catt|}

\vspace{-3pt}
\noindent
{\scriptsize\verb|examples/example_eh_cattsu.catt|}

\vspace{2pt}
\noindent
The \textit{Eckmann-Hilton move} of higher category theory is an algebraic phenomenon that arises in the following context:
$$(x : \star) (a,b : \i_0\sub{x} \to \i_0\sub{x})$$
The claim is that one can construct a composite 3-cell of the following type:
$$\textsf{EH}_3\sub{a,b} : \comp_2 \sub{a,b} \to \comp_2 \sub{b,a}$$
The geometrical intuition is that $a,b$ ``braid'' around each other in the plane.

This Eckmann-Hilton 3-cell has been previously construc-ted in $\Catt$, and we link to its formalization above. It is a geometrically complex proof, involving substantial use of the weak unit structure, requiring in total 1574 coherence constructors in the proof term.

Given the substantial use of the unit structure, we would anticipate the proof would simplify considerably in $\Catt_\su$ when normalized. This is indeed what we observe, with the $\Catt_\su$ normal form involving only 27 coherence constructors in the proof term, a 50-fold reduction in complexity.

Links to both the $\Catt$ and $\Catt_\su$ formalizations are provided at the top of this example. The $\Catt_\su$ normal form is so short that we are able to provide it explicitly as Appendix~\ref{app:eh}.
\end{example}

\begin{example}[Syllepsis]\label{ex:syl}\

\vspace{0pt}
\noindent
{\scriptsize\verb|examples/example_syllepsis.catt|}

\vspace{2pt}
\noindent
The  Eckmann-Hilton move ``braids'' the generators $a,b$ past each other. After some consideration it becomes clear that there should be two homotopically distinct ways to achieve this, which could be called the ``braid'' and ``inverse braid'', and we would not expect these proofs to be equivalent.

However, if we increase the level of degeneracy in the types, these moves in fact become equivalent. To show this, we must instead use the following context:
$$(x : \star) (\alpha , \beta : \i_1\sub{\i_0\sub{x}} \to \i_1\sub{\i_0\sub{x}})$$
Here $\alpha,\beta$ are 3\-dimensional generators. We can adapt our definition of the \textsf{EH} cell to this higher dimension, yielding the following 4-cell:
\begin{align*}
\mathsf{EH}_4 \sub{\alpha, \beta} &: \comp_3 \sub{\alpha, \beta} \to \comp_3 \sub{\beta, \alpha}
\end{align*}
The syllepsis is then a 5-cell  of the following type:
$$\mathsf{SY} \sub{\alpha, \beta} : \mathsf{EH}_4 \sub{\alpha, \beta} \to \mathsf{EH}_4^{-1} \sub{\beta, \alpha}$$
The syllepsis 5\-cell has not been constructed in $\Catt$. However, we have successfully formalized it in $\Catt_\su$ due to the reduction in proof complexity in that theory. The formalization is linked above.

We believe this to be the first formalization of the syllepsis as a pure path-theoretic object. Following dissemination of our results, others have now also formalized the syllepsis in other theories~\cite{Syllepsis}.
\end{example}

\section{Disc trivialization}

In the following section, we prove the following structure theorem: in a disc context~$\uD^n$, up to definitional equality, every term is either a variable, or an iterated identity on a variable. So if we restrict to terms that use \textit{all} variables of $\uD^n$ (that is, terms which do not factor through a smaller context), then there is exactly \textit{one} definitional equivalence class of term in each dimension $k \geq n$, the normal form of which is the locally-maximal variable $d_n,$ or an iterated identity on this. In this sense, $\Catt_\su$ trivializes disc contexts.

This provides an interesting point of comparison with work of Batanin, Cisinski and Weber \cite{Batanin2013}, discussed further in Section~\ref{sec:relatedwork}, where strictly unital $\infty$\-categories are defined as algebras over a certain  operad, defined to be the universal one which trivializes the operations over discs, and which satisfies a unit compatibility property. By showing our theory also trivializes  operations over discs, we claim that in this respect, our theory is at least as strict.

\paragraph{Structure Theorem} Our structure theorem is stated as follows:

\begin{restatable}[Disc trivialization]{theorem}{disctrivialization}
\label{thm:disctrivialization}
Suppose $t$ is valid in $\uD^n$. Then $t$ is definitionally equal to a variable, or to the iterated canonical identity on a variable.
\end{restatable}

To prove this, we need to establish some technical results about pasting contexts. We say that a substitution $\Delta \vdash \sigma : \Gamma$ is a \emph{variable-to-variable substitution} if for all $x \in \var(\Gamma)$, we have that $x \sub \sigma$ is again a variable.

\begin{restatable}{lemma}{variablestovariables}
\label{lem:variablestovariables}
Let $\uD^n\vdash \sigma: \Gamma $ be a valid substitution which is in normal form, and which sends locally-maximal variables of $\Gamma$ to variables of $\uD^n$. Then $\sigma$ is a variable-to-variable substitution.
\end{restatable}
\begin{proof}
Let $v$ be a variable of $\Gamma$. Then there is some locally-maximal variable $w$ of $\Gamma$ such that $v = \src^k(w)$ for some $k \in \N$. It follows from the formation rules for substitution that
$\src^k(w \sub \sigma) = \src^k(w) \sub \sigma \equiv v \sub \sigma$. Since we are given that $w \sub \sigma$ is a variable, it follows that $v \sub \sigma$ is a variable up to definitional equivalence. But since $\sigma$ is normalized, $v \sub \sigma$ must be precisely a variable.
\end{proof}

For $k<n$, there are two variable-to-variable substitutions \mbox{$\uD^n \vdash \partial^{\pm}_{k,n}: \uD^k$}, which map the $k$\-disc context into the appropriate source or target context of~$\uD^n$. We also have $\uD^n \vdash \id_{\uD^n} : \uD^n$, the identity substitution. We call these \emph{subdisc inclusions.} We now show that every valid variable-to-variable substitution \mbox{$\uD^n \vdash \sigma : \Gamma$} is of this form.
\begin{restatable}{lemma}{pastingfact}
\label{lem:pastingfact}
Let $\uD^n \vdash \sigma:\Gamma$ be a valid variable-to-variable substitution. Then we have $\Gamma \equiv \uD^k$ for some $k \leq n$, and $\sigma$ is a subdisc inclusion.
\end{restatable}
\begin{proof}
\def\g{\mathsf{g}}
The variables of a pasting context $\Gamma$ form a globular set $\g(\Gamma)$ in an obvious way, and the substitution well-typedness condition means that a variable-to-variable substitution $\uD^n \vdash \sigma : \Gamma$ induces a function of globular sets $\sigma : \g(\Gamma) \to \g(\uD^n)$. Suppose for a contradiction $\Gamma$ is \textit{not} a disc: then it must contain some sub--Dyck word $(\Uparrow (\Downarrow (\Uparrow (\cdots )\,y\,f)\cdots ) z\,g \cdots)$, and we have $\tgt(f) = \src(g)$. Then also $\tgt(\sigma(f)) = \src(\sigma(g))$; but now we have a contradiction, since the globular set of a disc does not have any pair of elements related in this way.

We conclude that $\Gamma \equiv \uD^k$ for some $k \leq n$. It remains to show that $\sigma$ is a subdisc inclusion. For this, suppose $k=n$. Then since $\sigma$ preserves variable dimension, we must have $d_n \sub \sigma = d_n$, and this extends uniquely to the other variables,  since $\sigma$ is a function of globular sets, and we conclude $\sigma \equiv \id_{\uD^n}$. Otherwise, suppose $k < n$. Then we can choose $d_k \sub \sigma = d_k$ or $d_k \sub \sigma = d'_k$, and once again, both extend uniquely, yielding $\sigma \equiv \partial _{k,n}^-$ and $\sigma \equiv \partial _{k,n}^+$ respectively.
\end{proof}

Given a valid term $t$ in some context $\Gamma$, its \emph{canonical identity} is $\i(t) := \i_{\dim \, \ty(t)} \sub{\{\ty(t),t\}}$. Canonical identities can be distinguished from ordinary identities $\i_n \sub \sigma$ because we do not need to give the dimension subscript, as it can be inferred from the term and the supplied context; because we use round brackets; and because we supply a term as an argument, rather than a substitution. A term is an \emph{iterated canonical identity} if it is of the form $\i^k(t)$, by applying this construction $k$ times for $k > 0$. We now show that if a term is definitionally equal to an ordinary identity $\i_n \sub \sigma$, it is definitionally equal to a canonical identity.

\begin{restatable}{lemma}{canonicalid}
\label{lem:canonicalid}
If $t$ is a valid term of $\Gamma$ with $t = \i_n \sub {\ldots , p}$, then $t = \i(p)$.
\end{restatable}
\begin{proof}
We define $\sigma := [\ldots , p_2^{}, p_2', p_1^{}, p'_1, p]$. Because $\Gamma \vdash \sigma : \uD^n$ is a valid substitution, it must satisfy the substitution typing conditions up to definitional equality, so we conclude for each $0 < k \leq n$ the following:
\begin{align*}
p_k \equiv d_{n-k} \sub \sigma \equiv \src^k(d_n) \sub \sigma = \src^k(d_n \sub \sigma) \equiv \src^k(p)
\\
p'_k \equiv d'_{n-k} \sub \sigma \equiv \tgt^k(d_n) \sub \sigma = \tgt^k(d_n \sub \sigma) \equiv \tgt ^k (p)
\end{align*}
We now reason as follows:
\begin{align*}
t &\equiv \i_n \sub{\ldots,  p_2, p_2', p_1, p_1', p }
\\
& = \i_n \sub{\ldots, \src^2(p), \tgt^2(p), \src(p), \tgt(p),p }
\\
& \equiv \i(p)
\end{align*}
This completes the proof.
\end{proof}

We are now able to prove Theorem~\ref{thm:disctrivialization}.

\begin{proof}[Proof of Theorem~\ref{thm:disctrivialization}]
If $t$ is a variable, we are done. Otherwise, $t$ is a coherence term, and we have that \mbox{$t\equiv \coh(\Gamma:U)[\sigma]$}. By Corollary~\ref{cor:decidedefeq}, we may assume without loss of generality that $t$ is in normal form.

If $t$ is an identity, then by Lemma~\ref{lem:canonicalid} we know $t = \i(u)$. By induction on dimension, $u$ is therefore either a variable or a iterated  identity on a variable, and we are done.

It remains to consider the case that $t\equiv \coh(\Gamma:U)[\sigma]$ is not an identity. We will see that this leads to a contradiction. Since $t$ is in normal form, we know that $t$ is not an A-, B-, C-, D- or E-redex, and we use these facts freely below.

First, note that $\uD^n \vdash \sigma:\Gamma$ maps locally maximal variables of $\Gamma$ to non-identity terms of $\uD^n$ (or else $t$ would be a B-redex), and these terms are in normal form (or else $t$ would be an A\-redex). Hence, by induction on subterms, we may assume that $\sigma$ maps locally-maximal variables to variables. By Lemma~\ref{lem:variablestovariables}, it follows that $\sigma$ is a variable-to-variable substitution, and then from Lemma~\ref{lem:pastingfact} we conclude that $\Gamma$ is a disc context $\uD^k$ with $k\leq n$, and $\sigma$ is a subdisc inclusion. We therefore conclude that $t \equiv \coh (\uD^k : u \to_T v)[\sigma]$.

Suppose $\dim \, t = k$. Then $u,v$ must each use all the variables of the respective boundary context, so by induction on subterms, the only possibility is  $u = d_{k-1}$ and $v = d'_{k-1}$. Since $t$ is not a C-redex, we conclude that $u$ is in normal form (hence $u \equiv d_{k-1}$), $v$ is in normal form (hence $v \equiv d'_{k-1}$), and $T$\ is in normal form (hence $T \equiv S^{k-2}$), and so $t \equiv \coh(\uD^k : d_{k-1} \to_{S^{k-2}} d'_{k-1})[\sigma]$. But then $t$ would be a D-redex, which is a contradiction.

So we must have $\dim \, t > k$.  Then $u,v$ must each use all the variables of $\uD^k$, so by induction on subterms, the only possibility is $u = v = \i^{n-k}(d_k)$. But this would mean that $t$ is an E-redex, again giving a contradiction.
\end{proof}

\section{Rehydration}
\label{sec:rehydration}

In this section, we show that for every $\Cattsupd$ term $t$, we can
produce a $\Catt^\pd$ term $R(N(t))$, its \emph{rehydrated normal
  form}, with the property that $t = R(N(t))$. To see why this is
useful, consider that while $N(t)$ is in normal form, its sources and
targets will not necessarily be. The rehydrated normal form ``fixes
up'' the boundaries recursively, putting them all into rehydrated
normal form. This means that if $u,v$ are composable in $\Cattsupd$,
the terms $R(N(u)),R(N(v))$ are again composable on-the-nose in
$\Catt^\pd$, and indeed are themselves valid in $\Catt^\pd$.

Our strategy is to introduce the following operations on any valid term $t$, simultaneously by mutual recursion:
\begin{itemize}
\item the \textit{rehydration} $R(t)$, which rehydrates all subterms, and then pads the resulting term;
\item the \textit{padding} $P(t)$, which composes a term at its boundaries, ensuring all of its sources and targets are in rehydrated normal form;
\item the \textit{normalizer} $\phi(t)$, a coherence term which provides an explicit equivalence  between $t$ and its rehydrated normal form $R(N(t))$.
\end{itemize}

We begin with the definition of rehydration.
\begin{definition}
For a valid $\Cattsupd$ term $\Gamma \vdash t : A$, its \emph{rehydration} $R(t)$ is defined as follows:
\begin{itemize}
\item $R(x) := x$
\item $R(\coh ( \Theta : U)[\sigma] \big) := P \big( \coh(\Theta:R(U))[R(\sigma)] \big)$
\end{itemize}
On valid types and substitution, rehydration is defined by applying term rehydration to all subterms and subtypes.
\end{definition}

\noindent

We next give the definition of padding, which composes a term with normalizers to put all its boundaries into rehydrated normal form. Since these boundaries have strictly smaller dimension than original term, the mutual recursion  is well-founded. We use the notation $\mathsf{comp_{n,k}}$ for the coherence introduced on page~\pageref{page:comp}, and for simplicity list only the locally-maximal arguments of the substitution.

\begin{definition}
\label{def:padding}
For a valid $\Cattsupd$ term $\Gamma \vdash t : A$, its \emph{padding} $P(t)$ is defined by $P(t) := P_{\dim \, A + 1}(t)$, and then:
\begin{itemize}
\item $P_0(t) := t$
\item $\begin{aligned}[t]
    P_{k+1}(t) := \mathsf{comp}_{\dim \, A + 1, k} \big(&\phi(\src_k(P_k(t))),\\
    &P_{k}(t),\\
    &\phi^{-1}(\tgt_k(P_k(t))) \big)
\end{aligned}
$
\end{itemize}
\end{definition}

\noindent
The constructors $P_k$ each ``fix up'' the corresponding source and targets of their arguments, so $P_k(t)$ is guaranteed to have its $j$-boundary in rehydrated normal form for all $j < k$.

Finally we give the definition of the normalizer data.
\begin{definition}
\label{def:normalizer}
For a valid $\Cattsupd$ term $\Gamma \vdash t : A$,  its \emph{normalizer} $\phi(t)$ and \emph{inverse normalizer} $\phi^{-1}(t)$ are defined as:
\begin{align*}
\phi(t) &:= \coh \big( \Gamma: R(N(t)) \to_{\ty(t)} t \big) [\id_\Gamma]
\\
\phi^{-1}(t) &:= \coh \big( \Gamma: t \to_{\ty(t)} R(N(t))\big) [\id_\Gamma]
\end{align*}
\end{definition}

We now arrive at the main result of this section.

\begin{restatable}{theorem}{rntcatt}
  \label{thm:rntcatt}
  Let $\Gamma$ be a pasting context and suppose we have
  $\Gamma \vdash t : A$ for some $\Cattsupd$ term $t$ such that
  $\supp(t)=\FV(\Gamma)$.  Then the following properties hold:
  \begin{enumerate}
  \item $R(N(t))$ is a valid term in $\Catt$;
  \item $R(N(t)) = t$ in $\Catt_\su$.
  \end{enumerate}
\end{restatable}

\noindent
We will prove (1) and (2) by simultaneous induction on the structure
of terms as well as their dimension.

\begin{proof}[Proof of Property 1]
  It will be sufficient to prove that $R(t)$ is valid in $\Catt$ for
  any $t$ which is in normal form, since the statement clearly follows
  from this assertion.  We therfore proceed by induction on the
  structure of $t$ assuming that it (and therefore all of its subterms)
  are in normal form.

  If $t$ is a variable then $R(t) = t$ for which the statement is
  clearly true. Otherwise $t = \coh(\Theta:U)[\sigma]$ and we have by
  definition
  \[ R(t) = P(\coh(\Theta:R(U))[R(\sigma)]) \]
  We will first show that the subterm
  \[t' := \coh(\Theta:R(U))[R(\sigma)]\] is itself valid in $\Catt$
  and then argue that applying the padding operation $P$ preserves
  this validity. We first note that since a straightforward induction
  on the definition gives that $\supp(R(t)) = \supp(t)$, it follows
  that $R(U)$ will itself satisfy the free variable condition for
  $\coh$ term formation and that, by the induction hypothesis, $R(U)$
  is a valid \Catt type.

  Now, the terms comprising the substituion $R(\sigma)$ are
  themselves also valid by induction.  This is not, in itself, enough
  to ensure that $R(\sigma)$ is a valid \Catt substitution. The
  additional conditions which must be verified are all of the form
  \[ \src_k(R(u)) \equiv \tgt_l(R(v)) \] for inferred sources and
  targets of some terms $u, v \in \sigma$.  But in any such case, we
  must have $ \src_k(u) = \tgt_l(v) $ in $\Catt_\su$ by the assumption
  that the original term $t$ was valid. It follows that
  \[ R(N(\src_k(u))) \equiv R(N(\tgt_l(v))) \] using that $u$ and $v$
  are in normal form.

  We now claim that we have the equations
  \begin{align*}
    \src_k(R(u)) &\equiv R(N(\src_k(u)))  \\
    \tgt_l(R(u)) &\equiv R(N(\tgt_k(v)))
  \end{align*}
  which, combined with the observation above will complete the claim
  that $R(\sigma)$ is a valid $\Catt$ substitution.

  To see this, we calculate:
  \begin{align*}
    \src(R(u)) &\equiv \src \Big(
                 \begin{aligned}[t]
                   \comp_{\dim \, A +1, \dim A} \big(&\phi(\src(P_{\dim A}(u))),\\
                   &P_{\dim A}(N(t)),\\
                   &\phi^{-1}(\tgt(P_{\dim A}(u))) \big) \Big)
\end{aligned}
    \\
              &\equiv \src (\phi (\src(P_{\dim A}(u)))) \\
               &\equiv \src
                 \begin{aligned}[t]
(\coh (\partial^-(\Gamma) {}:{} &R(N(\src(P_{\dim A}(u))))\\ \to{} &\src(P_{\dim A}(u)) )[\id_\Gamma])
\end{aligned}
    \\
               &\equiv R(N(\src(P_{\dim A}(u)))) \\
              &\equiv R(N(\src(N(P_{\dim A}(u))))) \\
                &\equiv R(N(\src(u)))
  \end{align*}
  This second to last step follows since
$ N(\src(N(t))) \equiv N(\src(t))$ as a consequence of definitional equality preserving typing. The last step follows by an application
  of property (2) of the theorem, which is valid in this case because $\src(u)$
  must be of strictly smaller dimension.  The result for $\src_k$
  follows by a simple induction on $k$.  The case for $\tgt_k$ and $v$
  is similar.

  We have now shown that $t'$ is indeed a valid $\Catt$ term and we
  note for below that it also follows that all of its inferred sources
  and targets are valid $\Catt$ terms as well.

  We now argue that $P(t')$ is valid in $\Catt$.  We prove this by
  induction on the parameter $k$ of the padding construction:
  \begin{align*}
    P_{k+1}(t') \equiv \mathsf{comp_{\dim \, A + 1, k}} \big(&\phi(\src_k(P_k(t'))),\\
                                                             &P_{k}(t'),\\
                                                             &\phi^{-1}(\tgt_k(P_k(t'))) \big)
  \end{align*}
  The composition operations $\mathsf{comp_{d,k}}$ are certainly valid
  \Catt terms, as elementary syntactic constructions. By induction,
  the subterm $P_k(t')$ is valid in \Catt. The normalizers in this
  expression are being computed for terms
  $\src_k(P_k(t'))$ and $\tgt_k(P_k(t'))$.  While these terms are not
  themselves in rehydrated normal form, their inferred types
  $\ty(\src_k(P_k(t')))$, $\ty(\tgt_k(P_k(t')))$ \textit{will} be in
  rehydrated normal form, thanks to the remark following
  Definition~\ref{def:padding}.

  To complete the proof, we must therefore show that if $u$ is some
  valid term with $\ty(u)$ in rehydrated normal form, then $\phi(u)$
  and $\phi^{-1}(u)$ are valid in $\Catt$. We consider the definition
  of $\phi(u)$:
  $$\phi(u) := \coh \big( \Gamma: R(N(u)) \to_{\ty(u)} u \big) [\id_\Gamma]$$
  Since in all cases of interest $\dim \, u < \dim \, t$, we know by
  induction on dimension that $R(N(u))$ is valid in \Catt. We know $u$
  arises as an inferred source or target of $t'$, and hence is valid
  in \Catt. We also know $\ty(u)$ is already in rehydrated normal
  form, and hence $\ty(u) \equiv \ty(R(N(u)))$.  So $\phi(u)$ is valid
  in \Catt, as is~$\phi^{-1}(u)$.
\end{proof}

\begin{proof}[Proof of Property 2]
  For a variable, this is immediate, and we therefore focus on
  coherence terms.  We may assume that the result holds on all terms
  $u$ of strictly smaller dimension than~$t$. This allows us to show
  that such terms have normalizers which themselves normalize to
  identities:
  \begin{align*}
    \phi(u) &\equiv \coh (\Gamma : u \to_{\ty(u)} R(N(u)))[\id_\Gamma] \\
            &\xsquig C _{\rt} \coh(\Gamma: N(u) \to_{N(\ty(u))} N(u))[\id_\Gamma] \\
            & \xsquig E \i_{\dim \, t} \sub{\{N(\ty(u)), N(u)\}}
  \end{align*}
  Note that we use the equation $N(R(N(u))) \equiv N(u)$ in the first step,
  this being a consequence of the conclusion of the theorem applied to
  $u$, which is of smaller dimension than $t$.

  Now, since paddings are constructed from the normalizers of terms of
  strictly smaller dimension (namely the sources and targets of $t$),
  we may use the previous result to show that paddings normalize to to
  the term being padded.  This we prove by induction on the parameter
  $k$ in the definition of the padding composite $P_k$:
  \begin{align*}
    P_{k+1}(t) &\equiv \mathsf{comp_{n, k}}(\phi(\src_k(t)), P_k(t), \phi(\tgt_k(t))) \\
               &\xsquig A _\t \mathsf{comp_{n,k}}(\i_k \sub{\tau }, t, \i_k \sub{\tau }) \\
               &\xsquig B_\t \coh(D^n:d_n \to_{S^{n-1}}d_n)[t] \\*
               &\xsquig D t
  \end{align*}
  where $n:=\dim \, t$.

  To finish the claim, we argue by induction on the structure of
  $t$. As in the proof of (1), we may as well suppose that $t$ is in
  normal form.  Moreover, we may also use the fact that $R(t)$ is
  valid in $\Catt$ and therfore in $\Catt_\su$ to promote the
  reductions described above to definitional equalities. Then
  (ignoring the trivial case of variables) we have:
  \begin{align*}
    R(t) &\equiv R(\coh ( \Theta : U)[\sigma] \big) \\
         &\equiv P \big( \coh(\Theta:R(U))[R(\sigma)] \big) \\
         &=  \coh(\Theta:R(U))[R(\sigma)] \\
         &= \coh(\Theta:U)[\sigma]
  \end{align*}
  This completes the proof.
\end{proof}

\noindent
With Theorem~\ref{thm:rntcatt} in hand, we are able to conclude with the result promised
in Section~\ref{sec:models}.

\begin{restatable}{corollary}{strictunit}
The functor $K^*:\cCat^{\mathrm{su}}_{\infty} \to \cCat_{\infty}$ is fully faithful.
\end{restatable}
\begin{proof}
  It is a well-known elementary observation (cf~\cite[p.47]{BS}) that
  if a functor $F: \cat{A} \to \cat{B}$ is essentially surjective on
  objects and full, then for any category $\cat{C}$, the induced
  precomposition functor
  $ - \circ F: \Hom(\cat{B}, \cat{C}) \to \Hom(\cat{A}, \cat{C})$
  between functor categories is fully faithful.

  Now, we have already seen that the functor $K$ is the identity on
  objects.  Moreover, the argument in the proof of
  Theorem~\ref{thm:rntcatt}, shows that for any $\sigma$ which is a
  valid substitution in $\Catt_\su$, $R(N(\sigma))$ is a valid
  substitution in $\Catt$ and moreover that $R(N(\sigma)) = \sigma$.
  This shows that the functor $K$ is full.

  As $\cCat^{\mathrm{su}}_{\infty}$ and $\cCat_{\infty}$ are
  full subcategories of the presheaf categories
  $\Hom((\Cattsupd)^\op, \cat{Set})$ and
  $\Hom((\Catt^\pd)^\op, \cat{Set})$, respectively, and since
  precomposition with the full functor $K: \Catt^\pd \to \Cattsupd$
  preserves these full subcategories, the result follows.
\end{proof}

\noindent
In other words, it is a \emph{property} of a given $\infty$-category
to be strictly unital; if there exists a lift of a given
$\infty$-category $\cat{C} \in \cCat_{\infty}$ to a strictly unital
one $\cat{C}' \in \cCat^{\mathrm{su}}_{\infty}$, it is necessarily unique up to isomorphism, and hence a weak $\infty$\-category can admit at most one strictly
unital structure.

\newpage
\bibliographystyle{plainurl}
\bibliography{references_arxiv}

\newpage
\appendix

\section{The Eckmann-Hilton term}
\label{app:eh}

Here we give the normalized  Eckmann-Hilton term in $\Catt_\su$, as discussed in Example~\ref{ex:eh}. This is available at the following link:

\vspace{3pt}

\noindent
{\scriptsize\url{http://github.com/ericfinster/catt.io/tree/v0.1}}

\vspace{-3pt}
\noindent
{\scriptsize\verb|examples/example_4_3_cattsu.catt|}

\vspace{5pt}
\noindent
For readability, we omit implicit arguments in the following.

{\scriptsize
\begin{align*}
  &\coh((x : \star) (y : \star) (f : x \shortrightarrow y) (g : x \shortrightarrow y) (\alpha : f \shortrightarrow g)\\
  &\hspace{5.4em}(z : \star) (h : y \shortrightarrow z) (i : y \shortrightarrow z) (\beta : h \shortrightarrow i) :\\
  &\quad\coh((x : \star) (y : \star) (f : x \shortrightarrow y) (g : x \shortrightarrow y) (\alpha : f \shortrightarrow g)\\
  &\hspace{9.6em}(h : x \shortrightarrow y) (\beta : g \shortrightarrow h) : f \shortrightarrow h)[\\
  &\quad\quad\coh((x : \star) (y : \star) (f : x \shortrightarrow y) (g : x \shortrightarrow y) (\alpha : f \shortrightarrow g) (z : \star) (h : y \shortrightarrow z) :\\
  &\quad\quad\quad\coh((x : \star) (y : \star) (f : x \shortrightarrow y) (z : \star) (g : y \shortrightarrow z) : x \shortrightarrow z)[f,h]\\
  &\quad\quad\quad\shortrightarrow\\
  &\quad\quad\quad\coh((x : \star) (y : \star) (f : x \shortrightarrow y) (z : \star) (g : y \shortrightarrow z) : x \shortrightarrow z)[g,h])[\alpha,h],\\
  &\quad\quad\coh((x : \star) (y : \star) (f : x \shortrightarrow y) (z : \star) (g : y \shortrightarrow z) (h : y \shortrightarrow z) (\alpha : g \shortrightarrow h):\\
  &\quad\quad\quad\coh((x : \star) (y : \star) (f : x \shortrightarrow y) (z : \star) (g : y \shortrightarrow z) : x \shortrightarrow z)[f,g]\\
  &\quad\quad\quad\shortrightarrow\\
  &\quad\quad\quad\coh((x : \star) (y : \star) (f : x \shortrightarrow y) (z : \star) (g : y \shortrightarrow z) : x \shortrightarrow z)[f,h])[g,\beta]\\
  &\quad]\\
  &\quad\shortrightarrow\\
  &\quad\coh((x : \star) (y : \star) (f : x \shortrightarrow y) (g : x \shortrightarrow y) (\alpha : f \shortrightarrow g)\\
  &\hspace{9.6em}(h : x \shortrightarrow y) (\beta : g \shortrightarrow h) : f \shortrightarrow h)[\\
  &\quad\quad\coh((x : \star) (y : \star) (f : x \shortrightarrow y) (z : \star) (g : y \shortrightarrow z) (h : y \shortrightarrow z) (\alpha : g \shortrightarrow h) :\\
  &\quad\quad\quad\coh((x : \star) (y : \star) (f : x \shortrightarrow y) (z : \star) (g : y \shortrightarrow z) : x \shortrightarrow z)[f,g]\\
  &\quad\quad\quad\shortrightarrow\\
  &\quad\quad\quad\coh((x : \star) (y : \star) (f : x \shortrightarrow y) (z : \star) (g : y \shortrightarrow z) : x \shortrightarrow z)[f,h])[f,\beta],\\
  &\quad\quad\coh((x : \star) (y : \star) (f : x \shortrightarrow y) (g : x \shortrightarrow y) (\alpha : f \shortrightarrow g) (z : \star) (h : y \shortrightarrow z) :\\
  &\quad\quad\quad\coh((x : \star) (y : \star) (f : x \shortrightarrow y) (z : \star) (g : y \shortrightarrow z) : x \shortrightarrow z)[f,h]\\
  &\quad\quad\quad\shortrightarrow\\
  &\quad\quad\quad\coh((x : \star) (y : \star) (f : x \shortrightarrow y) (z : \star) (g : y \shortrightarrow z) : x \shortrightarrow z)[g,h])[\alpha,i]\\
  &\quad])[a,b]
\end{align*}
}

\end{document}